\documentclass[aps, pra, twocolumn, superscriptaddress, floatfix, nofootinbib, longbibliography]{revtex4-2}

\usepackage{amsmath,amssymb,amsfonts}
\usepackage{amsthm}
\usepackage{braket}
\usepackage{graphicx}
\usepackage[table]{xcolor}
\usepackage{booktabs}       
\usepackage{array}
\usepackage{hyperref}
\usepackage{bm}

\arrayrulecolor{gray}

\hypersetup{
  pdftitle ={Contextuality from the Projector Overlap Matrix},
  pdfauthor={A. C. G\"unhan, S. S. Aksoy, Z. Gedik},
  pdfkeywords={quantum contextuality, projector overlap matrix,
               Kochen-Specker, KCBS, CHSH},
  pdfsubject={Quantum Information}
}

\newtheorem{proposition}{Proposition}

\newcommand{\MIE}{E}
\newcommand{\Tcal}{\mathcal{T}}
\newcommand{\Gcal}{\mathcal{G}}
\newcommand{\HS}{\mathrm{HS}}

\newcommand{\MU}{\mathrm{MU}}

\newcommand{\tr}{\mathrm{Tr}}
\newcommand{\CF}{\mathrm{CF}}

\begin{document}

\title{Contextuality from the Projector Overlap Matrix}

\author{A. C. G\"unhan}
\email{alicangunhan@mersin.edu.tr}
\affiliation{Department of Physics, Mersin University,
  \c{C}iftlikk\"oy Merkez Yerle\c{s}kesi,
  Yeni\c{s}ehir, Mersin, 33343 T\"urkiye}

\author{S. S. Aksoy}
\affiliation{\c{S}ehit Mustafa G\"oksal Anatolian High School,
  Ministry of National Education, 01230 Adana, T\"urkiye}

\author{Z. Gedik}
\affiliation{Faculty of Engineering and Natural Sciences,
  Sabanc\i\ University, Tuzla, 34956 \.Istanbul, T\"urkiye}

\date{\today}

\begin{abstract}
We place several known indicators of Kochen--Specker
contextuality---the KCBS correlator $\chi$, the contextual
fraction $\CF$, the Shannon-entropic $n$-cycle inequality of
Chaves and Fritz, and the operational commutator witness
$D$ of Paper~I---into a single projector-geometric
framework organized around the overlap matrix
$\Tcal_{ij} = d^{-1}\tr[(\hat P_i \hat Q_j)^2]$, where
$\hat P_i$ and $\hat Q_j$ are the joint-eigenspace
projectors of the two compatible observable pairs within a
measurement context. The state-independent scalar content
of $\Tcal$ is carried by two independent contractions: the
mutual information energy $E = \sum_{ij}\Tcal_{ij}$ of
Paper~I (equivalently, its logarithmic form
$S_2 = -\log_2 E$), and the Maassen--Uffink extremal
overlap $c_\MU = \max_{i,j}|\langle a_i,b_i | c_j,b_j\rangle|$.
We prove that $S_2$ is non-increasing under coarse-graining,
that $S_2(\Gcal) > 0$ is a necessary configuration-level
condition for observable contextuality, and that the
additive composition
$S_2(\Gcal) = \sum_\alpha S_2(\Gcal_\alpha)$ is exact for the KCBS pentagon. We further show that in the spin-$1$ realization of the KCBS pentagon, a shared $m_s=0$ eigenstate in each context forces $c_\MU = 1$, rendering every Maassen--Uffink-type bound trivial---a structural mechanism that makes explicit why outcome-entropic uncertainty relations based on $c_\MU$ are silent on KCBS contextuality, while $S_2 \approx 2.7266$~bits throughout. Applied to KCBS and CHSH, the framework identifies regimes in which every state-dependent witness considered here is silent yet $S_2(\Gcal) > 0$ by an amount set by the projector geometry alone.
\end{abstract}

\maketitle

\section{Introduction}
\label{sec:intro}

Kochen--Specker contextuality~\cite{KochenSpecker1967,CabelloRev22}
is the statement that the statistics of compatible
measurements in a quantum system cannot be reproduced by
any noncontextual hidden-variable model in which
measurement outcomes are determined by the system alone,
independent of the measurement context. Equivalently, the
empirical model on the family of contexts admits no global
joint distribution consistent with all contextual marginals,
in the sense of the classical extension problem of
Vorob'ev~\cite{Vorobev1962,AbramskyBrandenburger2011}. A range of
quantitative witnesses has been developed to test this
obstruction: outcome correlators such as the
Klyachko--Can--Binicio\u{g}lu--Shumovsky
inequality~\cite{KCBS} (whose pentagon construction
originates in Klyachko's earlier geometric-invariant-theory
analysis of coherent states~\cite{Klyachko2002}), the contextual
fraction from the sheaf-theoretic
framework~\cite{AbramskyBrandenburger2011,AbramskyBarbosaMansfield2017},
Shannon-entropic $n$-cycle
inequalities~\cite{ChavesFritz2012,Chaves2013,Kurzynski2012},
and resource-theoretic
monotones~\cite{AmaralNCW2018,Amaral2019RT}, among others.
Each of these witness is \emph{state-dependent}: it reports a value that depends on the specific quantum state prepared on the configuration. A complementary perspective is offered by \emph{state-independent} quantities that depend only on the measurement configuration itself, and whose positivity encodes geometric properties of the projector families---the Maassen--Uffink parameter $c_\MU$~\cite{MaassenUffink1988}, its quantum-memory extension~\cite{Berta2010,Coles2017}, and the mutual information energy $E$ introduced in Paper~I~\cite{GunhanGedik2026} all belong to this class. Graph-theoretic invariants of the compatibility (exclusivity) structure of a contextuality scenario---the Lov\'asz number and the fractional packing quantities~\cite{CabelloSeveriniWinter2014}---offer a related but distinct state-independent perspective: they are combinatorial invariants of the scenario graph that bound the noncontextual, quantum, and no-signaling correlation polytopes, whereas the scalar contractions of $\Tcal$ studied here are analytic invariants of the joint-eigenspace projector geometry.

Paper~I~\cite{GunhanGedik2026} introduced $E$ as a
projector-geometric quadratic contraction of the overlap
matrix $\Tcal_{ij} \equiv d^{-1}\tr[(\hat P_i \hat Q_j)^2]$,
and an operational quantity
$D(\Gcal_\alpha, \hat\rho) = |\langle [\hat A_\alpha,
\hat C_\alpha]\rangle_{\hat\rho}|$ which is
state-dependent but directly measurable through commutator
expectation values. The present paper places several known
contextuality indicators in a single framework organized
around $\Tcal$ and two independent scalar contractions: the
quadratic contraction $E$ and the extremal contraction
$c_\MU$. The logarithmic form
$S_2 \equiv -\log_2 E$ is the same quantity as $E$ written
to make additive composition across independent contexts
manifest; we use $S_2$ throughout as the working form and
reserve $E$ for its quadratic expression. $S_2$ records a
necessary configuration-level condition for observable
contextuality (Prop.~\ref{prop:necessary}): $S_2 > 0$
whenever any state-dependent witness fires, while remaining
positive also in regimes where every state-dependent
witness is silent. The extremal contraction
$c_\MU$---long established as the tight constant in the
Maassen--Uffink entropic uncertainty
relation~\cite{MaassenUffink1988} and its
extensions~\cite{Berta2010,Coles2017}---here enters in a
contextuality-analytic role:
Prop.~\ref{prop:shared_eigenstate} identifies a structural
mechanism forcing $c_\MU = 1$ in every context of the
spin-$1$ KCBS realization, rendering every $c_\MU$-based
entropic bound trivial in that scenario. Where the
operational quantity $D$ of Paper~I reads an
experimentally accessible signature of a particular state,
$S_2$ reads a state-independent signature of the
configuration; the two quantities complement rather than
replace one another.

The main results of this work are the following.
(i)~We place $E$ (equivalently $S_2 = -\log_2 E$) and
$c_\MU$ as two independent scalar contractions of $\Tcal$,
and establish the basic bound $E \le c_\MU^2$
(Sec.~\ref{sec:overlap}).
(ii)~We prove non-increase of $S_2$ under coarse-graining
of either projector family
(Prop.~\ref{prop:monotone}).
(iii)~We identify an exactness criterion for the additive
composition $S_2(\Gcal) = \sum_\alpha S_2(\Gcal_\alpha)$
and verify it for KCBS by cyclic orthogonality
(Sec.~\ref{sec:multi-context}).
(iv)~We establish that $S_2(\Gcal) > 0$ is a necessary
configuration-level condition for observable contextuality
(Prop.~\ref{prop:necessary}).
(v)~We prove that $c_\MU = 1$ in every context of the
spin-$1$ realization of the KCBS pentagon
(Prop.~\ref{prop:shared_eigenstate}), so that every
$c_\MU$-based entropic bound---Maassen--Uffink, its
memory extension, and the Shannon-entropic $n$-cycle
inequality of Chaves and Fritz---degenerates or remains
silent on the KCBS state family $\hat\rho(p)$, yet
$S_2(\Gcal_\mathrm{KCBS}) \approx 2.7266$~bits throughout
(Fig.~\ref{fig:kcbs_witnesses}).
(vi)~We examine the CHSH 4-cycle at two angle regimes and
record the complementary activation of $\chi_\mathrm{CHSH}$
(in its standard sign-flipped form) and the Chaves--Fritz
inequality identified in~\cite{Chaves2013}; the two regimes
realize opposite saturation behaviors of the bound
$E \le c_\MU^2$ within a single configuration
(Sec.~\ref{sec:CHSH_results}).

The organization is as follows.
Section~\ref{sec:setting} fixes the setting and states the
properties of the overlap matrix and its contractions.
Section~\ref{sec:indicators} lists the contextuality
indicators we compare: Table~\ref{tab:state_dep} collects
the state-dependent witnesses,
Table~\ref{tab:state_indep} the configuration-level
quantities, and Table~\ref{tab:landscape} summarizes the
scenario-by-scenario behavior.
Sections~\ref{sec:KCBS} and~\ref{sec:CHSH} work out the two
representative scenarios.
Section~\ref{sec:discussion} discusses the resulting
picture, Section~\ref{sec:outlook} lists open directions,
and Section~\ref{sec:conclusion} concludes. Computational
details for every table entry are collected in the
Supplemental Material~\cite{GunhanAksoyGedikSM}.

\section{Setting}
\label{sec:setting}

In what follows we work in finite-dimensional Hilbert space $\mathcal{H}$ of dimension $d\ge3$, the regime in which Kochen--Specker (KS) contextuality is non-degenerate. A measurement context, in the KS sense~\cite{KochenSpecker1967,CabelloRev22}, consists of three observables
\begin{equation}
  \Gcal_\alpha = \{\hat A_\alpha, \hat B_\alpha, \hat C_\alpha\},
  \label{eq:context}
\end{equation}
in which $\hat B_\alpha$ commutes with both $\hat A_\alpha$ and
$\hat C_\alpha$. The index $\alpha$ labels the context when
several are considered together, as in the
Klyachko--Can--Binicio\u{g}lu--Shumovsky (KCBS)
pentagon~\cite{KCBS}; within a single context it may be dropped.

Since $\hat A_\alpha$ and $\hat B_\alpha$ commute, they admit a common eigenbasis. Let $\hat P_i$ denote the projector onto the joint-eigenspace of $(\hat A_\alpha,\hat B_\alpha)$ with eigenvalue pair $(a_i, b_i)$, and let $\hat Q_j$ denote the projector onto the joint-eigenspace of $(\hat C_\alpha,\hat B_\alpha)$ with eigenvalue pair $(c_j, b_j)$. Both families are complete, $\sum_i \hat P_i = \sum_j \hat Q_j = \hat{\mathbb{I}}$, and need not be rank-$1$. The operational quantity of Paper~I,
\begin{equation}
  D(\Gcal_\alpha, \hat\rho)
    \equiv \bigl|\langle [\hat A_\alpha, \hat C_\alpha]\rangle_{\hat\rho}\bigr|,
  \label{eq:D_def}
\end{equation}
enters the comparison of state-dependent witnesses in
Sec.~\ref{sec:KCBS}--\ref{sec:CHSH}; the definitions of the present
section concern only the projector families $\{\hat P_i\}$ and
$\{\hat Q_j\}$.

\subsection{The overlap matrix and its contractions}
\label{sec:overlap}

The projector-geometric content of the pair
$(\{\hat P_i\},\{\hat Q_j\})$ is encoded in the overlap matrix
\begin{equation}
  \Tcal_{ij} \equiv \frac{1}{d}\,\tr\bigl[(\hat P_i \hat Q_j)^2\bigr],
  \label{eq:Tij_def}
\end{equation}
which depends only on the projectors and the dimension, not on
any quantum state. From the definition and completeness,
$0\le \Tcal_{ij}\le 1/d$ and
\begin{equation}
  \sum_{i,j}\Tcal_{ij} = \MIE(\Gcal_\alpha) \in [1/d,\,1],
  \label{eq:E_def}
\end{equation}
where $\MIE$ is the mutual information energy of Paper~I~\cite{GunhanGedik2026} defined inspired by Onicescu's information energy~\cite{onicescu1966theorie}; $\Tcal$ is not a probability table, as its total weight depends on the projector geometry and is generally less than unity. When all projectors commute,
$\MIE=1$; in the rank-$1$ mutually unbiased case,
$\MIE=1/d$~\cite{GunhanGedik2026}. In the rank-$1$ case,
$\Tcal_{ij} = d^{-1}|\langle a_i,b_i|c_j,b_j\rangle|^4$.

Two scalar contractions of $\Tcal$ arise naturally.
The quadratic contraction~\eqref{eq:E_def} is given in logarithmic form
\begin{equation}
  S_2 \equiv -\log_2 \MIE.
  \label{eq:S2_def}
\end{equation}
We use the base-$2$ logarithm throughout, so $S_2$, Shannon entropies $H(\cdot)$, and the Maassen--Uffink bound $-2\log_2 c_\MU$~\cite{MaassenUffink1988} are all reported in bits. This matches the conventions of
Refs.~\cite{ChavesFritz2012,Chaves2013,Kurzynski2012,Berta2010}; in the original Maassen--Uffink presentation~\cite{MaassenUffink1988} the bound appears in nats with $-2\ln c$. The extremal contraction is
\begin{equation}
  c_\MU \equiv \max_{i,j}\, |\langle a_i,b_i|c_j,b_j\rangle|,
  \label{eq:cMU}
\end{equation}
the maximal amplitude overlap that controls the Maassen--Uffink entropic uncertainty relation~\cite{MaassenUffink1988} and its quantum-memory extension~\cite{Berta2010,Coles2017}. For higher-rank projectors $c_\MU$ is defined as the maximal amplitude overlap between rank-$1$ constituents of the two eigenbases. 

These two contractions stand in a definite hierarchy. The matrix $\Tcal$ determines both $\MIE$ and $c_\MU$, but neither uniquely determines $\Tcal$, and $c_\MU$ does not determine $\MIE$. A quantitative bound between the two is
\begin{equation}
  \MIE \le c_\MU^2,
  \label{eq:E_leq_cMU2}
\end{equation}
which follows in the rank-$1$ case from $\sum_j|\langle a_i,b_i|c_j,b_j\rangle|^2=1$ by bounding one factor by $c_\MU^2$. The bound is saturated when every nonzero rank-$1$ overlap equals $c_\MU^2$. Paper~I~\cite{GunhanGedik2026} observed that this bound can be looser in some scenarios and tight in others; the two scenarios of Sec.~\ref{sec:KCBS}--\ref{sec:CHSH} make this dichotomy explicit: KCBS is strictly below saturation, while the CHSH 4-cycle at Bell-optimal angles saturates the bound exactly (Sec.~\ref{sec:CHSH_results}).

\subsection{Properties of \texorpdfstring{$S_2$}{S2}}
\label{sec:properties}

The logarithmic form $S_2 = -\log_2 \MIE$
is chosen to convert the multiplicative bound $\MIE \le c_\MU^2$ into the additive comparator $S_2 \ge 2\log_2(1/c_\MU)$ familiar from entropic uncertainty relations, while preserving the configuration-level character of $\MIE$. The resulting quantity has $S_2 \in [0, \log_2 d]$ and four further properties listed below.

\emph{Faithfulness.}~~$S_2=0$ iff $\MIE=1$, which holds iff every pair $(\hat P_i,\hat Q_j)$ commutes and the two families coincide up to relabeling.

\emph{Maximality.}~~$S_2=\log_2 d$ iff $\MIE=1/d$; saturated in the rank-$1$ mutually unbiased case.

\emph{State independence.}~~$S_2$ depends only on the spectral projectors of the observables; no quantum state appears.

\emph{Basis independence.}~~Under a unitary relabeling of $\mathcal{H}$, each entry of $\Tcal$ is unchanged by cyclicity of the trace.

A fifth property---non-increase under coarse-graining of either family---is the content of Sec.~\ref{sec:coarse}.

An identity established in Paper~I expresses $\MIE$ in terms of projector commutators. For any two Hermitian projectors $\hat P,\hat Q$ with $\hat P^2=\hat P$ and $\hat Q^2=\hat Q$,
\begin{equation}
  \|[\hat P,\hat Q]\|_{\HS}^2
    = 2\tr(\hat P\hat Q) - 2\tr\bigl[(\hat P\hat Q)^2\bigr],
  \label{eq:comm_HS}
\end{equation}
where $\|\hat X\|_{\HS}^2=\tr(\hat X^\dagger \hat X)$ is the squared Hilbert--Schmidt norm of $X$. Summing
over all pairs and using $\sum_i \tr(\hat P_i\hat Q_j)=d$,
\begin{eqnarray}
  1 - \MIE
    &=& \frac{1}{2d}\sum_{i,j}\|[\hat P_i,\hat Q_j]\|_{\HS}^2,\\
  S_2 &=& -\log_2(1 - \text{{\small RHS}}).
  \label{eq:S2_comm_form}
\end{eqnarray}
When all projectors commute, the commutator sum vanishes and $S_2=0$; as the sum grows, $S_2$ increases toward $\log_2 d$.

We use the word ``entropic'' sparingly and only to indicate a functional parallel: $S_2 = -\log_2 \MIE$ has the \emph{form} of a R\'enyi-$2$ collision entropy $-\log_2\sum_i p_i^2$, with $\MIE$ playing the role of $\sum_i p_i^2$~\cite{Renyi1961}. The analogy is structural rather than definitional: $\MIE$ is not generally a collision term of any probability distribution, and $S_2$ is therefore \emph{not} an entropy in the measure-theoretic sense. $S_2$ is a projector-geometric quantity whose logarithmic form is chosen for the reasons given after~\eqref{eq:E_leq_cMU2}, and we adopt ``entropic'' elsewhere in this paper only to refer to the Shannon-entropic inequalities of Refs.~\cite{ChavesFritz2012,Chaves2013,Kurzynski2012} and to the entropic uncertainty relations of Refs.~\cite{MaassenUffink1988,Berta2010,Coles2017}, never to $S_2$ itself.

\subsection{Coarse-graining non-increase}
\label{sec:coarse}

\emph{Coarse-graining} means merging two members of a family into their sum. Given $\{\hat P_i\}$ with orthogonal components, replacing a pair $\hat P_i,\hat P_{i'}$ by $\hat P_i+\hat P_{i'}$ gives a coarser family that is still a complete set of
orthogonal projectors.

\begin{proposition}
\label{prop:monotone}
Let $\{\hat P_i\}$ and $\{\hat Q_j\}$ be two complete families
of orthogonal projectors on $\mathcal{H}$, and let
$\{\hat P_k'\}$ denote the family obtained from $\{\hat P_i\}$
by replacing $\hat P_i,\hat P_{i'}$ (with $i\ne i'$) by their
sum. Then
\begin{equation}
  \MIE(\{\hat P_k'\},\{\hat Q_j\}) \ge \MIE(\{\hat P_i\},\{\hat Q_j\}),
  \label{eq:E_coarse}
\end{equation}
with equality iff $\hat P_i \hat Q_j \hat P_{i'} = 0$ for every
$j$. Consequently $S_2(\{\hat P_k'\},\{\hat Q_j\}) \le
S_2(\{\hat P_i\},\{\hat Q_j\})$, and the same holds for
coarse-graining of $\{\hat Q_j\}$.
\end{proposition}

The proof reduces to the inequality $\tr[\hat P_i\hat Q_j\hat P_{i'}\hat Q_j] =
\|\hat P_i\hat Q_j\hat P_{i'}\|_\HS^2 \ge 0$, which holds for $\hat P_i \hat P_{i'}=0$ (as required by orthogonality of the family) and any $\hat Q_j$. Summing over $j$ yields the non-negative change of $\MIE$ in~\eqref{eq:E_coarse}. The full derivation and numerical verification over random projector configurations are in the Supplemental Material~\cite{GunhanAksoyGedikSM}, Sec.~B.

Iterating Proposition~\ref{prop:monotone} covers every
coarse-graining of either family by successive pairwise merges.
Unitary relabelings leave every entry of $\Tcal$ unchanged, so
$S_2$ is invariant under them. The result is consistent with,
and narrower than, the resource-theoretic framework of
Amaral \textit{et al.}~\cite{AmaralNCW2018,Amaral2019RT}, in
which the free operations are noncontextual wirings---a
strictly larger class including pre-processing, composition,
and stochastic mixing. We have established non-increase of
$S_2$ under coarse-graining of the joint-eigenspace projector
families only; whether $S_2$ is non-increasing under the full
class of noncontextual wirings is left open. The word
``monotone'' is used here in this restricted sense.

\subsection{Multi-context composition}
\label{sec:multi-context}

For a collection of contexts $\Gcal = \{\Gcal_\alpha\}_{\alpha=1}^N$
we define the total quantity by summation,
\begin{equation}
  S_2(\Gcal) = \sum_{\alpha=1}^N S_2(\Gcal_\alpha)
             = -\log_2\prod_{\alpha=1}^N \MIE(\Gcal_\alpha).
  \label{eq:S2_multi}
\end{equation}
This additive composition is exact whenever the joint eigenbases across contexts either (i)~share no projectors, or (ii)~share only projector pairs whose overlap
contribution $\tr[(\hat P_i\hat Q_j)^2]$ vanishes. Both
regimes occur in the scenarios we analyze:
condition~(ii) is realized in the KCBS scenario of
Sec.~\ref{sec:KCBS}, where shared rank-$1$ projectors appear
as adjacent pairs
$(\ket{0_\alpha}\!\bra{0_\alpha},
\ket{0_{\alpha+1}}\!\bra{0_{\alpha+1}})$ whose contribution
vanishes by cyclic orthogonality
$\braket{0_\alpha|0_{\alpha+1}}=0$; condition~(i) is
realized in the CHSH scenario of Sec.~\ref{sec:CHSH}, where
the Alice/Bob tensor-product structure of the $4$-cycle
forces the shared observable $\hat B_\alpha=\hat A_\alpha$ to
live on a different subsystem in each consecutive context, so
that no ordered joint-projector pair appears in more than one
context. In both cases $\sum_\alpha S_2(\Gcal_\alpha)$
coincides with the deduplicated joint-overlap sum exactly.
Whether there exist KS-contextual scenarios for which the
additive composition strictly overcounts the joint-overlap
value is an open question; we return to it in
Sec.~\ref{sec:outlook}.

\subsection{\texorpdfstring{$S_2$}{S2} as a necessary configuration-level
  condition}
\label{sec:necessary}

The properties recorded so far yield an immediate but
useful consequence that we use throughout the
scenario-by-scenario analysis.

\begin{proposition}
\label{prop:necessary}
Let $\Gcal = \{\Gcal_\alpha\}_{\alpha=1}^N$ be a collection of contexts and let $S_2(\Gcal)$ be the additive value of~\eqref{eq:S2_multi}. If $S_2(\Gcal) = 0$, then every contextuality witness consistent with the empirical model on $\Gcal$ vanishes for every quantum state $\hat\rho$.
\end{proposition}

\begin{proof}
By the definition~\eqref{eq:S2_multi}, $S_2(\Gcal) = 0$ requires $S_2(\Gcal_\alpha) = 0$ for every $\alpha$ (since each $S_2(\Gcal_\alpha) \ge 0$), and by Faithfulness (Sec.~\ref{sec:properties}) this occurs iff every pair $(\hat P_i, \hat Q_j)$ within each context commutes, and the two families coincide up to relabeling. In that case each context $\Gcal_\alpha$ admits a joint eigenbasis: a hidden-variable assignment given by the corresponding spectral outcomes reproduces the quantum statistics on $\Gcal_\alpha$. If the joint eigenbases across contexts are consistent on shared observables (a condition already required for compatibility in the sense of Sec.~\ref{sec:setting}), a noncontextual model exists on $\Gcal$ as a whole; hence any contextuality witness vanishes for every $\hat\rho$.
\end{proof}

Equivalently, $S_2(\Gcal) > 0$ is a \emph{necessary}
configuration-level condition for $\Gcal$ to exhibit
observable contextuality on any quantum state. The converse
does not hold in general: $S_2(\Gcal) > 0$ is compatible
with every state-dependent witness being silent, as
realized in the KCBS scenario for $p \le p_\star$
(Sec.~\ref{sec:KCBS_results}).

\section{Indicators of contextuality}
\label{sec:indicators}

We list the contextuality indicators we compare, separated into state-dependent witnesses (Table~\ref{tab:state_dep}) and configuration-level quantities (Table~\ref{tab:state_indep}), and we summarize the scenario-by-scenario behavior in Table~\ref{tab:landscape}. The state-dependent witnesses are the cycle correlator $\chi$ (the KCBS sum for $n=5$ and the sign-flipped CHSH combination $\chi_\mathrm{CHSH}$ for $n=4$), the contextual fraction $\CF$ of Abramsky, Barbosa, and Mansfield~\cite{AbramskyBarbosaMansfield2017}, the Shannon-entropic $n$-cycle inequality $\mathrm{BC}_n^k$ of Chaves and Fritz~\cite{ChavesFritz2012,Chaves2013}, and the operational commutator witness $D$ of Paper~I~\cite{GunhanGedik2026}. The configuration-level quantities are the Maassen--Uffink extremal overlap $c_\MU$~\cite{MaassenUffink1988}, together with the associated Maassen--Uffink entropic bound $-2\log_2 c_\MU$ and its quantum-memory extension~\cite{Berta2010,Coles2017}; and the mutual information energy $E$ of Paper~I (equivalently, its logarithmic form $S_2 = -\log_2 E$). Every cell entry is derived in the Supplemental Material~\cite{GunhanAksoyGedikSM} (Secs.~D--E); values quoted below are exact where possible and rounded to four decimal places otherwise.

\begin{table*}[htb]
\caption{State-dependent witnesses of contextuality. KCBS values are on the family $\hat\rho(p)=p\ket{0_z}\!\bra{0_z}+(1-p)\hat{\mathbb{I}}/3$ with threshold $p_\star=(3\sqrt{5}+5)/20$. CHSH values are on Bell state $\ket{\Phi^+}=(\ket{00}+\ket{11})/\sqrt{2}$ at two angle choices: the Bell-optimal angles $(0,\pi/4,\pi/2,-\pi/4)$ (column CHSH-B) and the entropic-optimal angles of Ref.~\cite{Chaves2013} (column CHSH-E); see Sec.~\ref{sec:CHSH}. \textit{act.\ /\ sil.}\ records whether the witness reports quantum contextuality (active) or not (silent). The cycle correlator $\chi$ is $\sum_\alpha \langle\hat A_\alpha \hat A_{\alpha+1}\rangle$ for KCBS ($n=5$, no sign flip) and the sign-flipped CHSH combination $\chi_\mathrm{CHSH}$ for the CHSH columns ($n=4$, Eq.~\eqref{eq:CHSH_chi}). Structural reasons are discussed in the corresponding sections and in SM~\cite{GunhanAksoyGedikSM}.}
\label{tab:state_dep}
\begin{ruledtabular}
\begin{tabular}{l l l l}
Indicator & KCBS & CHSH-B & CHSH-E \\
\hline
$\CF$
  & $0$ if $p{\le}p_\star$; $2(\sqrt{5}{-}2)$ at $\ket{0_z}$
  & $\sqrt{2}-1$
  & $0$ (silent) \\
\hline
$\chi$ (cycle correlator)
  & act.\ iff $p{>}p_\star$
  & act.\ ($2\sqrt{2}$)
  & sil.\ ($1.5329$) \\
\hline
$\mathrm{BC}_n^k$ (Chaves--Fritz)
  & sil.; ${\le}{-}1.1667$ bits
  & sil.\ ($-1.2018$)
  & act.\ ($+0.2309$) \\
\hline
$D$ (Paper~I)
  & $0$ on $\hat\rho(p)$; act.\ on $\ket{+1_z}$
  & $0$ on $\ket{\Phi^+}$
  & $0$ on $\ket{\Phi^+}$ \\
\hline
$-2\log_2 c_\MU$
  & $=0$ (Prop.~3)
  & $1$~bit
  & ${\approx}\ 0.1002$~bits \\
Berta EUR-QSI
  & trivial
  & nontrivial
  & nontrivial \\
\end{tabular}
\end{ruledtabular}
\end{table*}

\begin{table*}[htb]
\caption{Configuration-level quantities, defined from the overlap matrix $\Tcal$. CHSH columns are the same two angle choices as in Table~\ref{tab:state_dep}.  ``Sat.''\ is the saturation test $E\stackrel{?}{=}c_\MU^2$. The last row indicates whether the additive composition $S_2(\Gcal) = \sum_\alpha S_2(\Gcal_\alpha)$ is exact (Sec.~\ref{sec:setting}); the mechanism is given in parentheses---``cyclic orth.''\ refers to mechanism~(ii) (every shared pair contributes zero by cyclic orthogonality), ``distinct bases''\ to mechanism~(i) (no shared pair).}
\label{tab:state_indep}
\begin{ruledtabular}
\begin{tabular}{l l l l}
Quantity & KCBS & CHSH-B & CHSH-E \\
\hline
$E$ (Paper~I)   & $0.6852$ per ctx. & $0.5000$ per ctx.     & $\{0.8147, 0.8748\}$ alternating     \\
\hline
$c_\MU$         & $1$ (Prop.~3)     & $1/\sqrt{2}$       & ${\approx}\ \{0.9469, 0.9659\}$  alternating        \\
\hline
$E=c_\MU^2$ Sat.& no                & yes                & no                      \\
\hline
$\bm{S_2(\Gcal)}$ \textit{(this work)}
                & $5(-\log_2 0.6852) \approx 2.7266$~bits
                & $4$~bits
                & ${\approx}\ 0.9770$~bits \\
\hline
additive composition
                & \textbf{exact} (cyclic orth.)
                & \textbf{exact} (distinct bases)
                & \textbf{exact} (distinct bases) \\
\end{tabular}
\end{ruledtabular}
\end{table*}

\begin{table*}[htb]
\caption{Scenario landscape. For each scenario (and for each
CHSH angle regime), the active / silent pattern of the four
main state-dependent witnesses and the configuration-level
value $S_2(\Gcal)$. The pattern visible in the last column is the main observation of this work: $S_2(\Gcal)>0$ uniformly, while
the silent cells in the state-dependent columns depend on the
scenario and the quantum state.}
\label{tab:landscape}
\begin{ruledtabular}
\begin{tabular}{l l l l l l}
Scenario & $\chi$ & $\mathrm{BC}_n^k$ & $\CF$ & $D$ & $S_2(\Gcal)$ \\
\hline
KCBS (spin-1)
  & act.\ iff $p>p_\star$
  & silent (Prop.~3)
  & act.\ iff $p>p_\star$
  & $0$ on $\hat\rho(p)$; act.\ on TR-broken $\hat\rho$
  & ${\approx}\ 2.7266$~bits \\
\hline
CHSH, B angles
  & active ($2\sqrt{2}$)
  & silent ($-1.2018$)
  & active ($\sqrt{2}{-}1$)
  & $0$ on $\ket{\Phi^+}$
  & $4$~bits \\
\hline
CHSH, E angles
  & silent ($1.5329$)
  & active ($+0.2309$)
  & silent ($0$)
  & $0$ on $\ket{\Phi^+}$
  & ${\approx}\ 0.9770$~bits \\
\end{tabular}
\end{ruledtabular}
\end{table*}

\section{KCBS scenario}
\label{sec:KCBS}

We examine the KCBS scenario~\cite{KCBS} in the spin-$1$ realization ($d=3$) and evaluate every indicator of Sec.~\ref{sec:indicators}. The state-dependent witnesses in Table~\ref{tab:state_dep} and the configuration-level quantities in Table~\ref{tab:state_indep} are read off the present computation; structural reasons are identified in Secs.~\ref{sec:KCBS_setup}--\ref{sec:KCBS_results}.

\subsection{Setup}
\label{sec:KCBS_setup}

Five dichotomic observables $\hat A_1,\ldots,\hat A_5$ on
$\mathcal{H}=\mathbb{C}^3$ are realized by~\cite{GunhanGedik2026}
\begin{equation}
  \hat A_\alpha
    = \hat{\mathbb{I}} - 2\ket{0_{\hat\ell_\alpha}}\!\bra{0_{\hat\ell_\alpha}},
  \qquad \alpha = 1,\ldots,5\ \pmod 5,
  \label{eq:A_KCBS}
\end{equation}
where $\ket{0_{\hat\ell_\alpha}}$ is the $m_s=0$ eigenstate of spin along direction $\hat\ell_\alpha$,
$\hat{\mathbf{S}}\cdot\hat\ell_\alpha$. The five unit vectors $\hat\ell_1,\ldots,\hat\ell_5$ lie on a cone of polar angle
\begin{equation}
  \theta_\mathrm{KCBS}
    = \arcsin\!\left(\frac{1}{\sqrt{2}\cos(\pi/10)}\right)
    \approx 48.0301^\circ,
  \label{eq:theta_KCBS}
\end{equation}
with azimuthal spacing $6\pi/5$, arranged so that
$\hat\ell_\alpha\cdot\hat\ell_{\alpha+1}=0$ for every $\alpha$ (mod~$5$) (Paper~I, Sec.~5). The angle $\gamma$ between next-nearest directions satisfies $\cos\gamma=(\sqrt{5}-1)/2$. The five overlapping contexts are
\begin{equation}
  \Gcal_\alpha
    = \{\hat A_{\alpha-1}, \hat A_\alpha, \hat A_{\alpha+1}\},
  \qquad \alpha = 1,\ldots,5\ \pmod 5,
\end{equation}
each with the structure of Sec.~\ref{sec:setting}: $\hat B_\alpha\equiv\hat A_\alpha$ is the shared observable commuting with its two outer partners.

Paper~I establishes in closed form the per-context value of $\MIE$~\cite{GunhanGedik2026} (see also SM~\cite{GunhanAksoyGedikSM}, Sec.~D.2):
\begin{equation}
  \MIE(\Gcal_\alpha)
    = \frac{3 - 4\cos^2\gamma + 4\cos^4\gamma}{3}
    \approx 0.6852,
  \label{eq:E_KCBS}
\end{equation}
uniform in $\alpha$ by pentagonal symmetry. Hence
\begin{equation}
  S_2(\Gcal_\alpha) = -\log_2 \MIE(\Gcal_\alpha) \approx 0.5453,
  \label{eq:S2_per_KCBS}
\end{equation}
and by the additive exactness criterion of Sec.~\ref{sec:multi-context} (derivation in SM~\cite{GunhanAksoyGedikSM}, Sec.~C)---the duplicated shared pairs
$(\ket{0_\alpha}\!\bra{0_\alpha},
\ket{0_{\alpha+1}}\!\bra{0_{\alpha+1}})$ contribute zero by cyclic orthogonality---the multi-context value is exact:
\begin{equation}
  S_2(\Gcal_\mathrm{KCBS})
    = 5\,S_2(\Gcal_\alpha)
    \approx 2.7266\ \text{bits}.
  \label{eq:S2_total_KCBS}
\end{equation}

\subsection{State family}
\label{sec:KCBS_state}

We probe the mixing family
\begin{equation}
  \hat\rho(p) = p\,\ket{0_z}\!\bra{0_z} + (1-p)\,\frac{\hat{\mathbb{I}}}{3},
  \qquad p \in [0,1],
  \label{eq:rho_p_KCBS}
\end{equation}
where $\ket{0_z}$ is the $m_s=0$ eigenstate along the symmetry axis (chosen to be the $z$-axis) of the KCBS cone. The family interpolates between the maximally mixed state ($p=0$) and the state that maximizes the standard KCBS violation ($p=1$). Since $S_2(\Gcal_\mathrm{KCBS})$ does not depend on $\hat\rho$, the parameter $p$ probes the state dependence of the four witnesses in Table~\ref{tab:state_dep} while leaving $S_2>0$ fixed.

The KCBS correlator is affine in $p$: $\chi(\hat\rho(0)) = -5/3$, and $\chi(\hat\rho(1)) \equiv \chi(\ket{0_z}) = 5 - 4\sqrt{5}$ (SM~\cite{GunhanAksoyGedikSM}, Sec.~D.5), so the noncontextual bound $\chi \ge -3$ is violated iff $p > p_\star$, with
\begin{equation}
  p_\star = \frac{3\sqrt{5}+5}{20} \approx 0.5854
  \label{eq:p_star}
\end{equation}
(derivation in SM~\cite{GunhanAksoyGedikSM}, Sec.~D.4). The contextual fraction~\cite{AbramskyBarbosaMansfield2017} vanishes for $p\le p_\star$ and agrees with the closed form $\CF(\ket{0_z}) = 2(\sqrt{5}-2)\approx 0.4721$ at $p=1$ (SM~\cite{GunhanAksoyGedikSM}, Sec.~D.6); by the general relation between $\chi$ and $\CF$ in Ref.~\cite{AbramskyBarbosaMansfield2017}, the onset of CF and the KCBS violation coincide at $p_\star$.

The operational quantity of Paper~I,
$D(\Gcal_\alpha,\hat\rho) =
|\langle[\hat A_{\alpha-1},\hat A_{\alpha+1}]\rangle_{\hat\rho}|$,
vanishes identically on $\hat\rho(p)$: the KCBS observables
are time-reversal even and every $\hat\rho(p)$ is time-reversal
invariant, so the commutator expectation is both real and
purely imaginary, hence zero. $D$ is nonzero on TR-broken
states; for instance,
$D(\ket{+1_z},\Gcal) \approx 6.4984$, about $67\%$ of the
global operator-norm bound $5\cdot 4\sqrt{\sqrt{5}-2} \approx
9.7174$ established in Paper~I (details in SM~\cite{GunhanAksoyGedikSM},
Sec.~D.8). Figure~\ref{fig:D_states} displays $D$ on a
one-parameter pure-state sweep interpolating
$\ket{0_z}$ to $\ket{+1_z}$ and on six representative
states.

\subsection{The shared-eigenstate mechanism}
\label{sec:shared_eigenstate}

The triviality of all $c_\MU$-based bounds in the KCBS scenario
has a structural origin that applies within a single context
and depends only on the commutation structure, not on the
quantum state.

\begin{proposition}
\label{prop:shared_eigenstate}
Let $\Gcal_\alpha = \{\hat A_{\alpha-1}, \hat A_\alpha,
\hat A_{\alpha+1}\}$ be a KCBS context as defined in
Sec.~\ref{sec:KCBS_setup}, and let $\ket{0_\alpha}$ denote
the $m_s=0$ eigenstate of $\hat{\mathbf{S}}\cdot\hat\ell_\alpha$. Then $\ket{0_\alpha}$ is a joint eigenstate of both $(\hat A_{\alpha-1},\hat A_\alpha)$
and $(\hat A_{\alpha+1},\hat A_\alpha)$, and consequently
$c_\MU(\Gcal_\alpha) = 1$.
\end{proposition}

\begin{proof}
By~\eqref{eq:A_KCBS}, $\hat A_\alpha\ket{0_\alpha}
= -\ket{0_\alpha}$. Since $\hat\ell_{\alpha-1}\cdot\hat\ell_\alpha
= \hat\ell_{\alpha+1}\cdot\hat\ell_\alpha = 0$, the states
$\ket{0_{\alpha\pm 1}}$ are orthogonal to $\ket{0_\alpha}$, and
hence $\hat A_{\alpha\pm 1}\ket{0_\alpha}=\ket{0_\alpha}$. Thus
$\ket{0_\alpha}\bra{0_\alpha}$ appears as the same rank-$1$ projector in
both joint eigenbases used to build $\Tcal$ for
$\Gcal_\alpha$. The corresponding amplitude overlap is
$|\braket{0_\alpha|0_\alpha}|=1$, which attains the maximum in the definition of $c_\MU$; hence $c_\MU(\Gcal_\alpha)=1$.
\end{proof}

As an immediate consequence, every entropic-uncertainty bound of the Maassen--Uffink form~\cite{MaassenUffink1988} and its quantum-memory extension by Berta \textit{et al.} \cite{Berta2010}, whose right-hand side is proportional to $-\log_2 c_\MU$, is trivially satisfied in the KCBS scenario for every quantum state: the bound degenerates to zero.

The Shannon-entropic $n$-cycle contextuality inequality of
Chaves and Fritz~\cite{ChavesFritz2012,Chaves2013},
\begin{equation}
  \mathrm{BC}_n^k
    = H(X_k X_{k+1}) + \!\!\sum_{j\neq k,k+1}\!\! H(X_j)
    - \sum_{j\neq k} H(X_j X_{j+1}) \le 0
  \label{eq:ChavesBC}
\end{equation}
for all $k$, is satisfied with strict margin everywhere along the family $\hat\rho(p)$: numerically
$\max_k \mathrm{BC}_5^k(\hat\rho(p))\le -1.1667$~bits (tabulation
in SM~\cite{GunhanAksoyGedikSM}, Sec.~D.7), the bound attained at
$\hat\rho=\ket{0_z}\!\bra{0_z}$ and further below it as $p$
decreases, more than one bit below the noncontextual
threshold. The same conclusion holds for the entropic
KCBS test introduced earlier by Kurzy\'nski \emph{et
al.}~\cite{Kurzynski2012}: entropic violation in the KCBS
scenario is known to require a specialized state and
measurement choice outside $\hat\rho(p)$, with maximum
reported violation $\approx 0.091$ bits.
In the spin-$1$ realization of the pentagon, the
two-dimensional multiplicity of the $+1$ eigenspace of each
$\hat A_\alpha$ and the cyclic orthogonality
$\braket{0_\alpha|0_{\alpha+1}}=0$ together keep both
inequalities well below saturation.

The mechanism---$c_\MU=1$ in every context---is not specific to $n=5$. The odd-$n$-cycle generalization of the KCBS scenario was introduced independently by Liang, Spekkens, and Wiseman~\cite{LiangSpekkensWiseman2011} and by Cabello, Severini, and Winter~\cite{CabelloSeveriniWinter2014}, with the corresponding noncontextuality inequalities shown to be tight in Ref.~\cite{AraujoQuintinoBudroniCunhaCabello2013}. In every odd-$n$-cycle realized in the same spin-$1$ form, adjacent vectors are orthogonal by construction, so Proposition~\ref{prop:shared_eigenstate} applies unchanged: $c_\MU(\Gcal_\alpha)=1$ in every context. We confirmed this numerically for $n\in\{5,7,9,11,13,15\}$ (SM~\cite{GunhanAksoyGedikSM}, Sec.~D.10).

\emph{Principal-angle structure.}~~By Paper~I~\cite[Eq.~(7)]{GunhanGedik2026}, each per-pair overlap admits the geometric decomposition $\Tcal_{ij} = d^{-1}\sum_\mu \cos^4\theta_\mu^{(ij)}$, where $\theta_\mu^{(ij)}$ are the principal angles between the two joint eigenspaces~\cite{BjorckGolub1973}. 

In the spin-1 realization, each $\hat A_\alpha=\hat{\mathbb{I}}-2\ket{0_\alpha}\bra{0_\alpha}$ has eigenvalue $-1$ on the one-dimensional span of $\ket{0_\alpha}$ and eigenvalue $+1$ on its two-dimensional orthogonal complement. Within a context, the central observable's $+1$-eigenspace is the same two-dimensional subspace for both neighbor pairs, but each neighbor decomposes it into a different pair of one-dimensional joint eigenspaces; cyclic orthogonality $\braket{0_\alpha|0_{\alpha+1}}=0$ additionally forces the $(-1,-1)$ joint eigenspace to be empty. Hence every joint-eigenspace projector is rank-1, and each pair $(\hat{P}_i,\hat{Q}_j)$ contributes a single principal angle; each context therefore yields nine angles in total. These take only four distinct values,
\begin{equation}
  \theta \in \Bigl\{0,\ \arccos\tfrac{1}{\sqrt{\varphi}},\
    \arccos\tfrac{1}{\varphi},\ \tfrac{\pi}{2}\Bigr\},
  \label{eq:KCBS_angles}
\end{equation}
with multiplicities $(1,2,2,4)$, where $\varphi=(1+\sqrt{5})/2$ is the golden ratio. The unique $\theta=0$ entry is the geometric signature of Proposition~\ref{prop:shared_eigenstate}---the shared $|0_\alpha\rangle$ eigenstate forces $c_\MU=1$---and the four $\theta=\pi/2$ entries register the orthogonal-subspace structure enforced by cyclic orthogonality $\braket{0_\alpha|0_{\alpha\pm 1}}=0$. The two intermediate angles are supplementary, $\arccos(1/\sqrt{\varphi}) + \arccos(1/\varphi) = \pi/2$, a structural consequence of the pentagonal geometry.

\subsection{Results}
\label{sec:KCBS_results}

\begin{figure*}[htb]
\includegraphics[width=\textwidth]{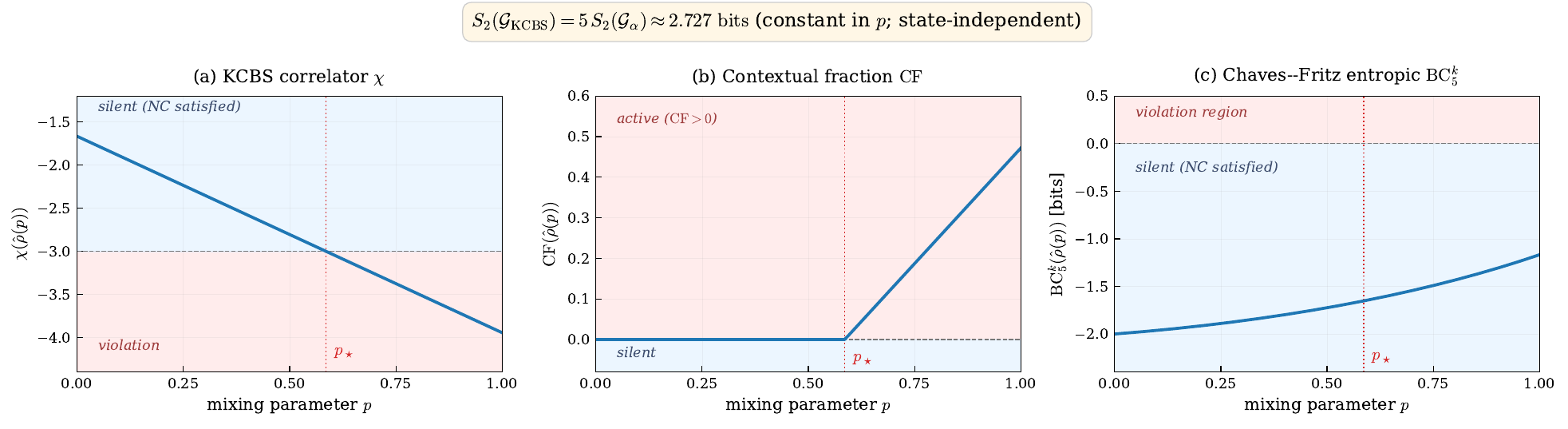}
\caption{Three of the four state-dependent witnesses of
Table~\ref{tab:state_dep} evaluated on the KCBS mixing family
$\hat\rho(p) = p\ket{0_z}\!\bra{0_z} + (1-p)\hat{\mathbb{I}}/3$,
with threshold $p_\star = (3\sqrt{5}+5)/20 \approx 0.5854$
(dotted vertical line in each panel).
(a)~The KCBS correlator $\chi$ is affine in $p$ and violates
the noncontextual bound $\chi \ge -3$ for $p > p_\star$.
(b)~The contextual fraction $\CF$ vanishes for
$p \le p_\star$ and rises to $2(\sqrt{5}-2) \approx 0.4721$
at $p=1$; its onset coincides with that of $\chi$, as
established in Ref.~\cite{AbramskyBarbosaMansfield2017}.
(c)~The Chaves--Fritz entropic
inequality~\eqref{eq:ChavesBC} is silent throughout,
$\mathrm{BC}_5^k \le -1.1667$~bits for all $p$, more than one
bit below the noncontextual threshold.
In the silent region $p \le p_\star$ (left of the dotted
vertical line), all three state-dependent witnesses are
simultaneously silent, while the configuration-level
quantity $S_2(\Gcal_\mathrm{KCBS}) \approx 2.7266$~bits
(header) is positive throughout, by an amount determined
by the projector geometry of the pentagon. The operational
witness $D(\Gcal_\mathrm{KCBS}, \hat\rho(p))$ of
Paper~I~\cite{GunhanGedik2026} vanishes identically on this
family by time-reversal invariance of $\hat\rho(p)$; its
behavior on time-reversal-broken states is displayed in
Fig.~\ref{fig:D_states}.}
\label{fig:kcbs_witnesses}
\end{figure*}

\begin{figure}[t]
\includegraphics[width=0.8\columnwidth]{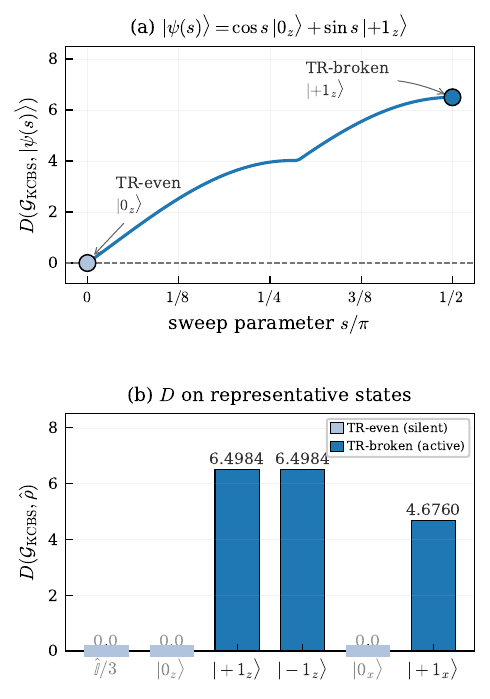}
\caption{Operational quantity $D(\Gcal_\mathrm{KCBS}, \hat\rho)$ of Paper~I~\cite{GunhanGedik2026} on states outside the time-reversal-invariant family $\hat\rho(p)$ of Fig.~\ref{fig:kcbs_witnesses}, where $D \equiv 0$. Throughout, ``TR'' abbreviates time reversal; a state is TR-even if it is invariant under the antiunitary time-reversal operation, and TR-broken otherwise. (a)~Pure-state sweep $\ket{\psi(s)} = \cos s\,\ket{0_z} + \sin s\,\ket{+1_z}$. The endpoints are qualitatively distinct: $s=0$ is the TR-even state $\ket{0_z}$ (open circle, $D = 0$) and $s=\pi/2$ is the TR-broken state $\ket{+1_z}$ (filled circle, $D \approx 6.4984$). Interior values $0 < s < \pi/2$ correspond to superpositions of the two endpoints and are neither purely TR-even nor purely TR-broken; $D$ rises across the sweep, with a mild non-monotonicity reflecting the absolute-value structure of $D$ as a sum of $|\langle[\hat A,\hat C]\rangle|$ terms.
(b)~$D$ on six representative states. $D$ vanishes on the maximally mixed state $\hat{\mathbb{I}}/3$ and on the TR-even spin-zero states $\ket{0_z}$ and $\ket{0_x}$; $D \approx 6.4984$ on $\ket{+1_z}$ and on $\ket{-1_z}$, and $D \approx 4.6760$ on $\ket{+1_x}$. The operational witness is sensitive to time-reversal breaking in a way complementary to $\chi$, $\CF$, and $\mathrm{BC}_5^k$, which all depend on properties that are invariant under time reversal.}
\label{fig:D_states}
\end{figure}

Figure~\ref{fig:kcbs_witnesses} displays three of the four
state-dependent witnesses on $\hat\rho(p)$, and
Fig.~\ref{fig:D_states} shows the operational quantity $D$
on a separate pure-state family. The separation is not
editorial: because $D \equiv 0$ on the TR-invariant family
$\hat\rho(p)$, displaying $D$ as a fourth panel of
Fig.~\ref{fig:kcbs_witnesses} would show a flat line at
zero and carry no information. Activating $D$ requires
moving off $\hat\rho(p)$ to TR-broken states. Conversely,
the three witnesses in Fig.~\ref{fig:kcbs_witnesses} vary
along the mixing parameter $p$ of the KCBS-maximizing
state family, and no single one-parameter state family in
the KCBS spin-$1$ realization activates all four witnesses
simultaneously: the Chaves--Fritz entropic inequality is
structurally silent on every quantum state in the spin-1 KCBS realization considered here by Prop.~\ref{prop:shared_eigenstate} (shared $m_s=0$ eigenstate), and $D$ requires TR-breaking that is
incompatible with the KCBS-maximizing $\ket{0_z}$ state at
which $\chi$ and $\CF$ are maximal. The two figures thus
display the same quantum scenario through two
complementary cuts through its state space.
Table~\ref{tab:state_dep} records all four in summary form:
$\chi$ and $\CF$ are active iff $p>p_\star$; the
Chaves--Fritz entropic inequality~\eqref{eq:ChavesBC} is
silent everywhere on $\hat\rho(p)$ with margin exceeding one
bit; and $D$ vanishes identically on $\hat\rho(p)$ by
TR-invariance. Table~\ref{tab:state_indep} records the
configuration-level entries:
$\MIE(\Gcal_\alpha) = (3-4\cos^2\gamma+4\cos^4\gamma)/3
\approx 0.6852$,
$c_\MU = 1$, non-saturation $\MIE < c_\MU^2$, and
$S_2(\Gcal) \approx 2.7266$~bits by the exact additive
composition.

Three features summarize the KCBS entries of
Table~\ref{tab:landscape}.

\textit{(a) Coincident onset of $\chi$ and $\CF$.}  Both become
nontrivial at $p=p_\star$, an instance of the general
correspondence between Bell-type violation and contextual
fraction established in
Ref.~\cite{AbramskyBarbosaMansfield2017}.

\textit{(b) Two witnesses are silent for distinct reasons.}
$\mathrm{BC}_5^k$ is silent because of projector geometry
(Prop.~\ref{prop:shared_eigenstate}); $D$ is silent because of state-and-observable symmetry (TR invariance of
$\hat\rho(p)$). Both silences are robust within their
respective classes: $\mathrm{BC}_5^k$ is structurally silent on every state in the spin-1 KCBS realization considered here (not only on $\hat\rho(p)$), while $D$ is
activated by a change of state (e.g., $\ket{+1_z}$) but is
not altered by varying $p$.

\textit{(c) $S_2$ is the only indicator that records contextuality throughout the mixing interval.}  For $p\le p_\star$ all four state-dependent witnesses are silent simultaneously, while $S_2(\Gcal_\mathrm{KCBS})\approx 2.7266$~bits---a property of the projector geometry, not of the state. This is the KCBS instance of the general observation articulated in Sec.~\ref{sec:discussion}: a state-dependent silence does not entail absence of contextuality at the configuration level.

\section{CHSH scenario}
\label{sec:CHSH}

The CHSH $4$-cycle~\cite{CHSH1969} occupies a distinct position among the scenarios considered here. Unlike KCBS, entropic violation of $\mathrm{BC}_n^k$ \emph{is} quantum mechanically achievable on $\ket{\Phi^+}$ at a suitable choice of measurement angles, as shown by Chaves~\cite{Chaves2013}; at the standard Bell-optimal angles, by contrast, $\mathrm{BC}_4^k$ is comfortably silent while $\chi_\mathrm{CHSH}$ reaches its Tsirelson value. We use CHSH as an instance in which two state-dependent witnesses, $\chi_\mathrm{CHSH}$ and $\mathrm{BC}_4^k$, exhibit a \emph{complementary} pattern of activation across measurement choices, and in which $S_2$ remains positive in both regimes.

\subsection{Setup}
\label{sec:CHSH_setup}

We take the $4$-cycle realization on two qubits in the Bell
state $\ket{\Phi^+}=(\ket{00}+\ket{11})/\sqrt{2}$, with
observables
\begin{equation}
\begin{aligned}
  \hat A_1 &= \hat\sigma_{a_0}\otimes\hat{\mathbb{I}}, \quad &
  \hat A_2 &= \hat{\mathbb{I}}\otimes\hat\sigma_{b_0}, \\
  \hat A_3 &= \hat\sigma_{a_1}\otimes\hat{\mathbb{I}}, \quad &
  \hat A_4 &= \hat{\mathbb{I}}\otimes\hat\sigma_{b_1},
\end{aligned}
\label{eq:A_CHSH}
\end{equation}
where $\hat\sigma_t = \cos(t)\hat\sigma_z + \sin(t)\hat\sigma_x$
and each $\hat\sigma_{a_i}$ ($\hat\sigma_{b_j}$) is Alice's
(Bob's) equatorial-plane spin observable at angle $a_i$
($b_j$), expressed throughout in radians. Adjacent pairs in
the cycle commute automatically, as each contains one Alice
operator and one Bob operator.
Each context $\Gcal_\alpha$ is
$\{\hat A_{\alpha-1}, \hat A_\alpha, \hat A_{\alpha+1}\}$
with $\hat B_\alpha = \hat A_\alpha$, following the same
convention as in Sec.~\ref{sec:setting}.

For the $n=4$ cycle, the Abramsky--Brandenburger--Kishida analysis~\cite{AbramskyBrandenburger2011} shows that the noncontextual bound on the cyclic sum $\sum_\alpha \langle \hat A_\alpha \hat A_{\alpha+1}\rangle$ is trivial and a sign flip on exactly one term is required to obtain a non-degenerate inequality. We accordingly use the standard CHSH form
\begin{equation}
  \chi_\mathrm{CHSH}
  \equiv \langle \hat A_1 \hat A_2\rangle
    + \langle \hat A_2 \hat A_3\rangle
    - \langle \hat A_3 \hat A_4\rangle
    + \langle \hat A_4 \hat A_1\rangle,
  \label{eq:CHSH_chi}
\end{equation}
with noncontextual bound $\chi_\mathrm{CHSH} \le 2$, quantum Tsirelson bound $\chi_\mathrm{CHSH} \le 2\sqrt{2}$~\cite{CHSH1969,Tsirelson1980}, and no-signaling maximum $\chi_\mathrm{CHSH} \le 4$~\cite{PopescuRohrlich1994}. This sign-flip convention is specific to even $n$; the KCBS 5-cycle of Sec.~\ref{sec:KCBS} requires no sign flip.

\subsection{Two angle regimes}
\label{sec:CHSH_regimes}

Two angle choices---both previously studied in the literature---serve as the representative extremes (see SM~\cite{GunhanAksoyGedikSM}, Sec.~E.2).

\textit{Bell-optimal angles.}~~
$(a_0, b_0, a_1, b_1) = (0,\,\pi/4,\,\pi/2,\,-\pi/4)$: Alice measures in directions $(0, \pi/2)$ and Bob in $(\pi/4, -\pi/4)$, the standard assignment that maximizes the CHSH correlator. On $\ket{\Phi^+}$ this gives $\chi_\mathrm{CHSH} = 2\sqrt{2}$, saturating the Tsirelson bound~\cite{Tsirelson1980}.

\textit{Entropic-optimal angles.}~~$(a_0, b_0, a_1, b_1) \approx (0,\,0.916,\,0.524,\,-2.880)$: the angle set identified in the analysis of Ref.~\cite{Chaves2013} as producing a quantum violation of the Chaves--Fritz entropic inequality~\eqref{eq:ChavesBC}. On $\ket{\Phi^+}$ this gives $\chi_\mathrm{CHSH} \approx 1.5329$ (within the
noncontextual/local bound $\le 2$) and $\mathrm{BC}_4^k \approx +0.2309$~bits (entropic violation).

Per-cell derivations for each regime are in SM~\cite{GunhanAksoyGedikSM},
Sec.~E.3. Table~\ref{tab:state_dep} records the
state-dependent witness values; Table~\ref{tab:state_indep}
records the configuration-level quantities.

\subsection{Results}
\label{sec:CHSH_results}

Three features summarize the CHSH entries of
Tables~\ref{tab:state_dep}--\ref{tab:landscape}.

\textit{(a) Complementary activation of $\chi_\mathrm{CHSH}$ 
and $\mathrm{BC}_4^k$.}  At Bell-optimal angles the CHSH 
correlator $\chi_\mathrm{CHSH}$ reaches the Tsirelson bound 
$2\sqrt{2}$ while $\mathrm{BC}_4^k=-1.2018$~bits, comfortably below the noncontextual bound; at entropic-optimal angles 
$\mathrm{BC}_4^k$ is active by $0.2309$~bits while 
$\chi_\mathrm{CHSH} \approx 1.5329$ retreats within the 
noncontextual bound $\chi_\mathrm{CHSH} \le 2$. This 
complementarity---identified by 
Chaves~\cite{Chaves2013}---reflects the fact that the 
maximally nonlocal distribution on the 4-cycle is 
entropically equivalent to a classically correlated 
distribution, and that entropic violation requires a 
distinct angle choice. The contextual 
fraction~\cite{AbramskyBarbosaMansfield2017} tracks 
$\chi_\mathrm{CHSH}$: $\CF = \sqrt{2}-1 \approx 0.4142$ at 
Bell-optimal angles and $\CF = 0$ at entropic-optimal angles 
(since $\chi_\mathrm{CHSH} < 2$ there).  As in the KCBS 
scenario, $\chi_\mathrm{CHSH}$ and $\CF$ encode the same 
empirical content; an entropic witness ($\mathrm{BC}_4^k$) 
is required to detect the entropic-optimal regime.

\textit{(b) $D$ vanishes on $\ket{\Phi^+}$ independent of angle.}~~The commutator $[\hat A_{\alpha-1},\hat A_{\alpha+1}]$ in each context is pure-imaginary Hermitian, and $\ket{\Phi^+}$ is invariant under the combined time-reversal of both qubits; hence
$\langle [\hat A_{\alpha-1},\hat A_{\alpha+1}]
\rangle_{\ket{\Phi^+}} = 0$ identically, in both regimes. This is the CHSH analogue of $D$'s silence on $\hat\rho(p)$ in the KCBS scenario; the structural origin is the same (TR-invariance).

\textit{(c) $S_2 > 0$ and the saturation dichotomy.}
$S_2(\Gcal_\mathrm{CHSH})$ is positive at both angle
choices: $4$~bits at Bell-optimal angles
and $\approx 0.9770$~bits at entropic-optimal angles. In the
Bell-optimal regime, the overlap matrix $\Tcal$ saturates
the bound $\MIE = c_\MU^2$~\eqref{eq:E_leq_cMU2}, with $\MIE = c_\MU^2 = 1/2$. In the entropic-optimal regime,
$\MIE(\Gcal_\alpha) < c_\MU(\Gcal_\alpha)^2$ in every context (e.g., $\MIE = 0.8748 < 0.9329 = c_\MU^2$ for $\Gcal_{2,4}$): the bound is not saturated. Thus the CHSH 4-cycle itself realizes both structural regimes---saturation and non-saturation of $E \le c_\MU^2$---within a single configuration, through the choice of measurement angles.

\textit{(d) Additive exactness by distinct-bases mechanism.}
The totals above decompose across the four contexts as additive sums. At Bell-optimal angles every context is
symmetric, $\MIE(\Gcal_\alpha) = 1/2$ for
$\alpha=1,\ldots,4$, giving
$S_2(\Gcal_\alpha) = 1$~bit and total $4$~bits. At entropic-optimal angles the four contexts split into two
tiers by the Alice/Bob alternation,
$\MIE(\Gcal_1) = \MIE(\Gcal_3) \approx 0.8147$ and
$\MIE(\Gcal_2) = \MIE(\Gcal_4) \approx 0.8748$, giving
per-context values $S_2(\Gcal_{1,3}) \approx 0.2956$~bits and
$S_2(\Gcal_{2,4}) \approx 0.1929$~bits, with total
$2(0.2956) + 2(0.1929) \approx 0.9770$~bits. The additive composition~\eqref{eq:S2_multi} is exact in
both regimes by condition~(i) of
Sec.~\ref{sec:multi-context}: each context $\Gcal_\alpha$ uses a distinct triple of angles $(a_i,b_j,a_k)$ or $(b_i,a_j,b_k)$ from the alternating Alice/Bob structure, so the joint-eigenspace projectors $\hat P_i, \hat Q_j$---built as tensor products of one Alice and one Bob eigenstate---differ from those of any other context. No ordered pair $(\hat P_i, \hat Q_j)$ recurs across contexts; mechanism~(i) applies. This is distinct from the KCBS mechanism,
where duplicated pairs do occur but contribute zero by
cyclic orthogonality; the CHSH 4-cycle instead avoids
duplication altogether.

The CHSH scenario thus realizes, on a single family of
observables, the same pattern observed across scenarios: every state-dependent witness in Table~\ref{tab:state_dep} (column-by-column) admits an angle regime at which it is silent---yet
$S_2(\Gcal_\mathrm{CHSH})$ remains positive throughout.

\section{Discussion}
\label{sec:discussion}

The results of
Sec.~\ref{sec:KCBS}--\ref{sec:CHSH} admit a compact reading
in terms of the overlap matrix $\Tcal$ and its scalar
contractions.

\subsection{Configuration-level versus state-dependent
  contextuality}
\label{sec:discussion_structure}

The three scenarios differ sharply at the level of the
state-dependent witnesses in Table~\ref{tab:state_dep}, but
agree uniformly at the configuration level: in each case
$S_2(\Gcal) > 0$, by a value determined by the projector
geometry alone.

\emph{KCBS} occupies the regime in which every
$c_\MU$-based witness is trivial: the shared eigenstate
$\ket{0_\alpha}$ (Prop.~\ref{prop:shared_eigenstate})
forces $c_\MU = 1$ in every context, so Maassen--Uffink,
Berta, and the Chaves--Fritz Shannon-entropic inequality
all degenerate or remain far from saturation. The
correlator $\chi$ and the contextual fraction $\CF$
become active only for $p > p_\star$; the operational
quantity $D$ vanishes identically on the mixing family
$\hat\rho(p)$ by TR-invariance. For $p \le p_\star$ every
state-dependent witness is silent, while
$S_2(\Gcal_\mathrm{KCBS}) \approx 2.7266$~bits is positive
by exact additive composition over the five contexts.

\emph{CHSH} realizes, within a single configuration, the
complementarity between the standard CHSH correlator and
the entropic $\mathrm{BC}_4^k$ inequality identified by
Chaves~\cite{Chaves2013}: at Bell-optimal angles $\chi_\mathrm{CHSH} = 2\sqrt{2}$ is active while $\mathrm{BC}_4^k$ is silent; at entropic-optimal angles these roles reverse. The contextual fraction tracks $\chi_\mathrm{CHSH}$ and is active in the Bell regime but silent in the entropic regime, so no state-dependent witness examined here is simultaneously active throughout the angle sweep. The two regimes also differ on the saturation $\MIE = c_\MU^2$: Bell-optimal angles saturate the bound exactly while entropic-optimal angles do not. CHSH thus realizes, within a single 4-cycle configuration, both structural regimes of the bound $E \le c_\MU^2$.

What unifies the two scenarios is the contrapositive of
Prop.~\ref{prop:necessary}: a nonzero configuration-level
$S_2$ is necessary for the existence of a quantum state
that activates any contextuality witness. The converse
is not available; both scenarios realize regimes in which
$S_2 > 0$ coexists with silent state-dependent witnesses
(KCBS at $p \le p_\star$; silent cells of
Table~\ref{tab:state_dep} for CHSH). In this sense $S_2$
is the weakest positive indicator: active whenever anything
else is, but also active in settings where nothing else is.

\subsection{\texorpdfstring{$D$}{D} versus \texorpdfstring{$S_2$}{S2}: operational commutator vs.\
  projector geometry}
\label{sec:discussion_D_S2}

The operational quantity $D$ of
Paper~I~\cite{GunhanGedik2026} and the configuration-level
$S_2$ of the present work complement each other in a
concrete way visible across the scenarios considered. Both
derive from commutator structure, but through different
ingredients: $D$ through expectation values of the outer
commutators $[\hat A_{\alpha-1}, \hat A_{\alpha+1}]$ on a
specific quantum state, and $S_2$ through the squared
Hilbert--Schmidt norms $\|[\hat P_i, \hat Q_j]\|_\HS^2$
of the joint-projector commutators, summed over all
projector pairs (Eq.~\eqref{eq:S2_comm_form}).

In KCBS, $D$ probes an operational signature that is
sensitive to state preparation but silent on the
TR-symmetric $\hat\rho(p)$ family; $S_2$ is insensitive
to state preparation entirely and reflects the pentagon
geometry. In CHSH, $D$ vanishes identically on
$\ket{\Phi^+}$ at both angle regimes by the same
TR-invariance mechanism, while $S_2$ discriminates between
the Bell-optimal and entropic-optimal regimes ($4$~bits
versus ${\approx}\ 0.9770$~bits). The two quantities are not
redundant: Paper~I's $D$ gives an experimentally measurable
quantity whose value depends on $\hat\rho$; Paper~II's
$S_2$ gives a state-independent classifier of the
configuration.

\subsection{Limitations}
\label{sec:discussion_limitations}

The framework presented here has bounds that we state explicitly.

\begin{enumerate}
    \item[(i)] $S_2$ is not the unique scalar contraction of
    $\Tcal$. The logarithmic form was chosen for its
    additive composition across independent contexts
    (Sec.~\ref{sec:multi-context}) and its clean commutator
    expression (Sec.~\ref{sec:properties}). Other
    contractions, such as $\ell_p$-norms and spectral
    invariants of $\Tcal$, are likely to probe aspects of the
    projector geometry not seen by $S_2$; a comparative
    study is outside the scope of this work.
    
    \item[(ii)] The additive composition~\eqref{eq:S2_multi} across
multiple contexts is verified in the KCBS scenario, for
the reasons stated in Sec.~\ref{sec:multi-context} and
verified in Sec.~\ref{sec:KCBS_setup}. Whether there
exist KS-contextual scenarios in which the additive
composition strictly overcounts the deduplicated
joint-overlap sum is an open question noted in
Sec.~\ref{sec:multi-context} and listed in
Sec.~\ref{sec:outlook}.

    \item[(iii)] The coarse-graining non-increase of Prop.~\ref{prop:monotone} establishes monotonicity of $S_2$ under a restricted class of operations: coarse-graining of the joint-eigenspace projector families. The broader class of noncontextual wirings~\cite{AmaralNCW2018,Amaral2019RT} includes transformations not covered by the present proof; whether $S_2$ is monotone under the full class remains open.

    \item[(iv)] The four state-dependent witnesses compared here ($\chi$, $\CF$, $\mathrm{BC}_n^k$, $D$) do not exhaust the landscape of contextuality measures. Noise-robust inequalities (e.g., the Contextuality-by-Default framework of Dzhafarov and Kujala~\cite{DzhafarovKujala2016}), memory-based uncertainty relations~\cite{Berta2010,Coles2017}, and distance-based contextuality measures~\cite{GrudkaHorodecki2014,AmaralNCW2018} all carry further information. The comparison carried out here is limited to the four witnesses listed in Sec.~\ref{sec:indicators}, chosen to span outcome-correlation ($\chi$), Shannon-entropic ($\mathrm{BC}_n^k$), linear-programming ($\CF$), and operational ($D$) perspectives.
    
    \item[(v)] The Chaves--Fritz inequality applies only to $n$-cycle configurations. State-independent contextuality scenarios outside the $n$-cycle family---such as the Peres--Mermin square or the Yu--Oh configuration---lie beyond the reach of both the Chaves--Fritz entropic witness and the present projector-geometric framework, which depends on the binary structure of a shared observable in each context. Extending entropic or projector-geometric witnesses to these configurations would require additional structure (e.g., algebraic parity obstructions or compatibility-graph data) beyond the scalar contractions of $\Tcal$.
\end{enumerate}

\section{Outlook}
\label{sec:outlook}

Five directions suggest themselves, in order of proximity to the analysis presented here.

\textit{(a) Other scalar contractions of $\Tcal$.}~~The
$\ell_p$-norms $\sum_{ij}\Tcal_{ij}^p$ for $p \neq 2$,
spectral invariants of $\Tcal$ treated as a matrix, and
quantities derived from row-normalizations of $\Tcal$ all
deserve study. Different choices may yield quantities
that either share the monotonicity established in
Sec.~\ref{sec:coarse} or diverge from it; the comparison
would refine the hierarchy of contractions sketched in
Sec.~\ref{sec:overlap}.

\textit{(b) Entropic uncertainty relations involving
$\MIE$.}~~The triviality of $c_\MU$ in the KCBS scenario
(Prop.~\ref{prop:shared_eigenstate}) removes the content
of every Maassen--Uffink-type bound. A Shannon-entropic
uncertainty relation in which one side is
expressed through $\MIE$ rather than $c_\MU$ would not
degenerate in this scenario and could produce nontrivial
bounds where the current framework gives none. The
extension to quantum side information in the manner of
Berta~\emph{et al.}~\cite{Berta2010} is a concrete target.

\textit{(c) The odd-$n$-cycle family.}~~The structural
argument of Prop.~\ref{prop:shared_eigenstate} yields
$c_\MU(\Gcal_\alpha) = 1$ in every context of every odd-$n$
cycle realized in the spin-$1$ form of
Sec.~\ref{sec:KCBS_setup}. An analytic description of
$S_2(\Gcal_n)$ at finite $n$ remains open, although the asymptotic behavior $S_2(\Gcal_n)\sim 1/n^2$ is straightforward (SM, Sec~D.10).

\textit{(d) Noncontextual wirings as free operations for
$S_2$.}~~Establishing or refuting monotonicity of $S_2$
under the full class of noncontextual
wirings~\cite{AmaralNCW2018,Amaral2019RT} would clarify
whether $S_2$ admits a resource-theoretic interpretation
beyond the coarse-graining restriction established here.
A related question concerns the relationship between
$S_2$ and the resource-theoretic monotones of
Ref.~\cite{Amaral2019RT}.

\textit{(e) Additive excess and state-independent
scenarios.}~~The additive
composition~\eqref{eq:S2_multi} is verified as exact in
KCBS by cyclic orthogonality
(Sec.~\ref{sec:multi-context}). Whether there exist
KS-contextual scenarios in which the additive sum
strictly overcounts the deduplicated joint-overlap sum
remains open. A natural setting for this question is the
class of state-independent contextuality scenarios, such
as the Peres--Mermin square~\cite{Peres1990,Mermin1993} or
the Yu--Oh configuration~\cite{YuOh2012}, where the KS
obstruction is an algebraic parity or coloring condition
organized on a combinatorial structure (grid, vector
system) rather than a state-dependent inequality. In such
scenarios the projector-geometric layer captured by
$S_2$ is expected to compose with a complementary,
configuration-combinatorial obstruction layer not seen
by scalar contractions of $\Tcal$ alone. A quantitative
two-layer description that combines the present
$\Tcal$-framework with this combinatorial layer is a
natural direction for future work.

\section{Conclusion}
\label{sec:conclusion}

We have placed the mutual information energy $E$ of
Paper~I~\cite{GunhanGedik2026}---equivalently its
logarithmic form $S_2 = -\log_2 E$---and the
Maassen--Uffink extremal overlap $c_\MU$ into a common
framework of scalar contractions of the overlap matrix
$\Tcal$, and shown that $S_2$ records a necessary
configuration-level condition for observable contextuality within this framework
(Prop.~\ref{prop:necessary}) while being insensitive to
the state. Across two representative
scenarios---KCBS and CHSH---every state-dependent witness
we compared ($\chi$, $\CF$, $\mathrm{BC}_n^k$, $D$) is
silent in at least one regime of measurement choice or
state, while $S_2(\Gcal)$ is positive in every scenario
by an amount determined by the projector geometry. The
distinction between configuration-level and state-dependent
contextuality, made explicit by this framework, identifies $S_2$ as the broadest positive indicator among those considered. We emphasize that $S_2$ is not a Shannon-type entropy and does not define an uncertainty relation; it is a projector-geometric quantity (cf.\ Sec.~\ref{sec:properties}) that diagnoses the geometric obstruction---the shared-eigenstate mechanism of Prop.~\ref{prop:shared_eigenstate}---responsible for the failure of $c_\MU$-based entropic uncertainty bounds in the KCBS scenario.

\begin{acknowledgments}
The authors are grateful to B. Cantürk for his valuable insights.
\end{acknowledgments}

\bibliographystyle{apsrev4-2}
\bibliography{References_merged}

\end{document}


\title{Supplemental Material for\texorpdfstring{\\}{}
``Contextuality from the Projector Overlap Matrix''}

\author{A. C. G\"unhan}
\email{alicangunhan@mersin.edu.tr}
\affiliation{Department of Physics, Mersin University,
  \c{C}iftlikk\"oy Merkez Yerle\c{s}kesi,
  Yeni\c{s}ehir, Mersin, 33343 T\"urkiye}

\author{S. S. Aksoy}
\affiliation{\c{S}ehit Mustafa G\"oksal Anatolian High School,
  Ministry of National Education, 01230 Adana, T\"urkiye}

\author{Z. Gedik}
\affiliation{Faculty of Engineering and Natural Sciences,
  Sabanc\i\ University, Tuzla, 34956 \.Istanbul, T\"urkiye}

\date{\today}

\maketitle


\section{Computational framework}
\label{SM:framework}

All numerical computations reported in the main text and in
Secs.~\ref{SM:KCBS}--\ref{SM:CHSH} follow the same
primitives, summarized here. Throughout this Supplemental
Material we use the context notation of the main
text: a measurement context is a set
$\Gcal_\alpha = \{\hat A_{\alpha-1}, \hat A_\alpha,
\hat A_{\alpha+1}\}$ in which the middle observable
$\hat A_\alpha$ commutes with each neighbor, so that the
context splits into two compatible pairs:
$(\hat A_{\alpha-1}, \hat A_\alpha)$ (the ``left'' pair) and
$(\hat A_{\alpha+1}, \hat A_\alpha)$ (the ``right'' pair),
both sharing the middle observable $\hat A_\alpha$.

\emph{Joint-eigenspace projectors.} For a pair of commuting
Hermitian observables $(\hat A, \hat B)$ on $\mathcal{H}$ of
dimension $d$, their joint eigenspaces are computed by
simultaneous diagonalization. In practice we diagonalize
the Hermitian combination $\hat A + \lambda \hat B$ with
$\lambda \in (0,1)$ chosen to lift eigenvalue degeneracies
of $\hat A$; the resulting eigenvectors are then grouped by
their common $(a,b)$-eigenvalue pair to form the
\emph{joint-eigenspace projectors} $\{\hat P_k\}_k$ onto the
joint eigenspaces. These projectors are complete,
$\sum_k \hat P_k = \hat{\mathbb{I}}$, and mutually
orthogonal, but need not all be rank-$1$: if a joint
eigenspace is $r$-dimensional, the corresponding
$\hat P_k$ has rank $r$. In the two-qubit CHSH
realization every joint eigenspace is one-dimensional, so
all $\hat P_k$ reduce to rank-$1$ outer products
$\ket{v_k}\!\bra{v_k}$; in the spin-$1$ KCBS realization
each $\hat A_\alpha$ has a two-dimensional $+1$-eigenspace
and a one-dimensional $-1$-eigenspace, so joint-eigenspace
projectors can have rank $1$ or $2$ depending on how
the neighbor observable splits the $+1$ eigenspace.

\emph{Overlap matrix.} Given the joint-eigenspace projectors
$\{\hat P_i\}_i$ of the left pair
$(\hat A_{\alpha-1}, \hat A_\alpha)$ and $\{\hat Q_j\}_j$ of
the right pair $(\hat A_{\alpha+1}, \hat A_\alpha)$, the
overlap matrix $\Tcal$ of context $\Gcal_\alpha$ is
assembled entrywise:
\begin{equation}
  \Tcal_{ij} = \frac{1}{d}\,\tr\bigl[(\hat P_i \hat Q_j)^2\bigr].
  \label{eq:SM_T_def}
\end{equation}
The contractions $\MIE = \sum_{ij}\Tcal_{ij}$ and
$S_2 = -\log_2 \MIE$ are computed by direct summation and
logarithm. The extremal contraction $c_\MU$ is computed by
$c_\MU = \max_{i,j} \|\hat P_i \hat Q_j\|_\infty$ (the
largest singular value of $\hat P_i \hat Q_j$), which
reduces to $\max_{i,j}|\braket{v_i|w_j}|$ when both
families are rank-$1$ with eigenvectors $\ket{v_i}$ and
$\ket{w_j}$.

\emph{Empirical joint probabilities.} For a commuting pair
$(\hat A, \hat B)$ with eigenvalue spectra $\pm 1$ and
eigenprojectors $\hat P^A_{\pm}, \hat P^B_{\pm}$, the joint
outcome distribution in state $\hat\rho$ is
$p(a, b) = \tr\bigl[\hat P^A_a \hat P^B_b \hat\rho\bigr]$
with the usual clamp at zero to handle floating-point
underflow. Marginal distributions follow by summing out.

\emph{Shannon entropies.} All Shannon entropies $H(\cdot)$
are computed as
$H(p) = -\sum_i p_i \log_2 p_i$ with convention
$0 \log_2 0 = 0$. Only contributions with $p_i > 10^{-14}$ are included, to eliminate floating-point noise.

\emph{Contextual fraction.} For scenarios where $\CF$ admits a closed-form expression (the $n$-cycle scenarios), we use the formula
$\CF(\hat\rho) = \max\bigl(0, (\chi_\mathrm{NC} -
\chi(\hat\rho))/(\chi_\mathrm{NC} - \chi_\mathrm{NS})\bigr)$
of Ref.~\cite{AbramskyBarbosaMansfield2017}, with
$\chi_\mathrm{NC}$ the noncontextual bound and
$\chi_\mathrm{NS}$ the no-signaling extremum. For the
KCBS $5$-cycle we use $\chi_\mathrm{NC} = -3$,
$\chi_\mathrm{NS} = -5$; for the CHSH $4$-cycle,
$\chi_\mathrm{NC} = 2$, $\chi_\mathrm{NS} = 4$.

\emph{Numerical conventions.} All floating-point
computations use complex double precision. Hermitian
matrices are enforced to be exactly Hermitian by
symmetrization (adding the conjugate transpose and dividing by $2$) before diagonalization, to prevent eigenvalue drift from rounding. Expectation values with
$|\cdot| < 10^{-9}$ are reported as zero in the main
text tables.

\section{Coarse-graining non-increase: proof and numerical check}
\label{SM:coarse}

This section proves Proposition~1 of the main text: $\MIE$
is non-decreasing, and $S_2 = -\log_2\MIE$ is non-increasing,
under any coarse-graining of either the left or the right
family of joint-eigenspace projectors. A coarse-graining here
means replacing two members of a family by their sum, i.e.\
merging two such projectors into a single larger one. The
proof reduces to a non-negativity statement about Hermitian
projectors (the Lemma below), which we then apply
termwise to obtain the proposition. We close with a
numerical check on random configurations.

\emph{Lemma.} For any Hermitian projector $\hat Q$
($\hat Q^2 = \hat Q = \hat Q^\dagger$) and any two
orthogonal Hermitian projectors $\hat P_1, \hat P_2$
(with $\hat P_1 \hat P_2 = 0$),
\begin{equation}
  \tr\bigl[\hat P_1 \hat Q \hat P_2 \hat Q\bigr] \ge 0.
  \label{eq:SM_lemma}
\end{equation}

\emph{Proof of the lemma.} Let $\hat X = \hat P_1 \hat Q
\hat P_2$. Then
\begin{equation}
  \hat X^\dagger \hat X
  = \hat P_2 \hat Q \hat P_1 \hat P_1 \hat Q \hat P_2
  = \hat P_2 \hat Q \hat P_1 \hat Q \hat P_2,
\end{equation}
using $\hat P_1^2 = \hat P_1$ and Hermiticity of each factor. The Hilbert--Schmidt norm squared is
\begin{equation}
  \|\hat X\|_\HS^2
  = \tr(\hat X^\dagger \hat X)
  = \tr\bigl[\hat P_2 \hat Q \hat P_1 \hat Q \hat P_2\bigr].
\end{equation}
Cycling $\hat P_2$ from the front to the back gives
$\tr\bigl[\hat Q \hat P_1 \hat Q \hat P_2\bigr]$, and one further cyclic shift yields
\begin{equation}
  \|\hat X\|_\HS^2
  = \tr\bigl[\hat P_1 \hat Q \hat P_2 \hat Q\bigr].
\end{equation}
Since $\|\hat X\|_\HS^2 \ge 0$, the claim follows. Equality holds iff $\hat P_1 \hat Q \hat P_2 = 0$, i.e.\ the sandwiched product vanishes as an operator.
\hfill$\square$

\emph{From the lemma to the proposition.} Let
$\{\hat P_i\}_i$ be the left family and
$\{\hat Q_j\}_j$ the right. Merging two members
$\hat P_i, \hat P_{i'}$ ($i \ne i'$) into
$\hat P_\text{merged} = \hat P_i + \hat P_{i'}$ gives a
coarser family $\{\hat P'_k\}_k$. The change in
$d\,\MIE = \sum_{ij} \tr[(\hat P_i \hat Q_j)^2]$ under
this single merge is
\begin{align}
  \Delta
  &= \sum_j \tr\bigl[(\hat P_i + \hat P_{i'})\hat Q_j
     (\hat P_i + \hat P_{i'})\hat Q_j\bigr] \notag\\
  &\quad
    - \sum_j \tr[(\hat P_i \hat Q_j)^2]
    - \sum_j \tr[(\hat P_{i'} \hat Q_j)^2].
  \label{eq:SM_merge_delta}
\end{align}
Expanding $(\hat P_i + \hat P_{i'})\hat Q_j (\hat P_i +
\hat P_{i'})\hat Q_j$ gives four terms: two diagonal terms that cancel the subtractions, and two cross terms,
$\tr[\hat P_i \hat Q_j \hat P_{i'} \hat Q_j]$ and
$\tr[\hat P_{i'} \hat Q_j \hat P_i \hat Q_j]$. These are equal by Hermiticity of $\hat Q_j$ and cyclicity of the trace. Therefore
\begin{equation}
  \Delta = 2\sum_j \tr\bigl[\hat P_i \hat Q_j
     \hat P_{i'} \hat Q_j\bigr] \ge 0,
  \label{eq:SM_delta_nonneg}
\end{equation}
the non-negativity following from the Lemma applied to each
term. Hence $d\,\MIE$ is non-decreasing under the merge,
$\MIE$ itself is non-decreasing (since $d > 0$ is fixed),
and $S_2 = -\log_2 \MIE$ is non-increasing. An identical
argument applies to merges in the right family
$\{\hat Q_j\}_j$.

Every coarse-graining of either family is expressible as a
sequence of pairwise merges, and the non-increase of $S_2$
composes: the full proposition follows.

\emph{Numerical verification.} We checked the proposition
on $10^4$ random configurations of Hermitian projectors
with $d \in \{3, 4, 5, 6\}$ and $n \in \{4, 5, 6, 8\}$
(left-family size), drawing pairs of Haar-random
rank-$1$ projectors and applying random pairwise merges.
In every instance $S_2$ either strictly decreased or
remained unchanged; exact equality occurred precisely when
$\hat P_i \hat Q_j \hat P_{i'} = 0$ for all $j$ (the
equality condition of the Lemma).

\section{Multi-context additive exactness}
\label{SM:exactness}

In a scenario with $N$ contexts $\Gcal = \{\Gcal_\alpha\}_{\alpha=1}^N$,
the main text defines the multi-context value
$S_2(\Gcal) = -\log_2 \prod_\alpha E(\Gcal_\alpha) = \sum_\alpha
S_2(\Gcal_\alpha)$ by the additive composition. This
definition is convenient because it reduces to a sum of
per-context calculations, but whether the sum equals the
quantity one would obtain by treating all contexts as a
single enlarged projector-pair collection is not automatic.
We develop the exactness criterion here, identify the two
mechanisms that make it hold, and verify both on KCBS
(mechanism ii, cyclic orthogonality) and CHSH (mechanism i,
distinct bases).

The additive composition
\begin{equation}
  S_2(\Gcal) = -\log_2 \prod_{\alpha=1}^N E(\Gcal_\alpha)
  = \sum_{\alpha=1}^N S_2(\Gcal_\alpha)
  \label{eq:SM_addcomp}
\end{equation}
is a sum of per-context values. It equals the value derived
from the union of all left-right ordered joint-projector
pairs across contexts exactly when no pair appears more
than once, or when every duplicated pair has zero overlap
contribution.

\emph{Deduplication.} Each context $\Gcal_\alpha$ generates
a set of ordered pairs
$\mathcal{E}_\alpha = \{(\hat P_i, \hat Q_j)\}_{i,j}$ where
$\{\hat P_i\}_i$ and $\{\hat Q_j\}_j$ are the left and
right joint-eigenspace projectors of $\Gcal_\alpha$. The combined
multiset $\bigcup_\alpha \mathcal{E}_\alpha$ may contain
the same ordered pair more than once when two contexts
share joint-eigenspace projectors (this happens, for instance,
when two contexts share an observable). The deduplicated
set $\mathcal{E}_\mathrm{dedup}$ lists each ordered pair
exactly once.

The contribution of an ordered pair $(\hat P, \hat Q)$ to
the overlap sum is $\tr[(\hat P \hat Q)^2]/d$. The
multi-context overlap $\MIE(\Gcal)$ therefore equals the
product over $\alpha$ of $\MIE(\Gcal_\alpha)$ iff
(i) no duplicates exist, or
(ii) every duplicated pair contributes zero.

\emph{KCBS: cyclic orthogonality (mechanism~ii).} In the
KCBS pentagon, adjacent contexts $\Gcal_{\alpha-1}$ and
$\Gcal_\alpha$ share the rank-$1$ projector
$\ket{0_\alpha}\!\bra{0_\alpha}$: this projector is in the
right joint-eigenspace projector family of $\Gcal_{\alpha-1}$ and
in both joint-eigenspace projector families of $\Gcal_\alpha$. Hence the ordered pair
$(\ket{0_\alpha}\!\bra{0_\alpha},
\ket{0_{\alpha+1}}\!\bra{0_{\alpha+1}})$ appears in both
the combined collections of $\Gcal_\alpha$ and
$\Gcal_{\alpha+1}$. But the pair's contribution is
$\tr\bigl[(\ket{0_\alpha}\!\bra{0_\alpha}
\ket{0_{\alpha+1}}\!\bra{0_{\alpha+1}})^2\bigr] = 0$, since
$\braket{0_\alpha | 0_{\alpha+1}} = 0$ by cyclic
orthogonality (\ref{SM:KCBS-setup}). The five such
duplicated pairs all contribute zero, and the additive
composition~\eqref{eq:SM_addcomp} is exact.

\emph{CHSH: distinct bases (mechanism~i).} In the CHSH
4-cycle, the alternating Alice/Bob tensor-product structure
gives each context a joint-eigenspace projector family on a
distinct tensor-factor pattern: $(\sigma_A, I_B)$ for one
and $(I_A, \sigma_B)$ for its neighbor. No rank-$1$
projector is shared between neighboring contexts, and the
full set of $4 \times (4 \times 4) = 64$ ordered pairs
across the four contexts contains no duplicate. Mechanism~(i) applies, and the additive composition is
exact in both the Bell-optimal and entropic-optimal angle
regimes.

\emph{Numerical verification.} For both scenarios we
verified the exactness computationally by (i) computing
$\MIE(\Gcal_\alpha)$ context by context and taking the
product, (ii) computing the multi-context overlap over the
deduplicated ordered-pair collection, and (iii) confirming
equality to within floating-point precision
($< 10^{-12}$). Table~\ref{tab:SM_exact_verify}
summarizes the counts.

\begin{table}[!htb]
\caption{Exactness verification. Total ordered
joint-projector pairs and duplicate count per scenario;
exactness mechanism (i or ii); resulting $S_2(\Gcal)$.}
\label{tab:SM_exact_verify}
\begin{ruledtabular}
\begin{tabular}{l c c c c}
Scenario & total pairs & duplicates & mech. & $S_2(\Gcal)$ \\
\colrule
KCBS (5 contexts)          & $45$ & $5$ (all zero)  & (ii) & $2.7266$ \\
CHSH Bell-optimal          & $64$ & $0$             & (i)  & $4.0000$ \\
CHSH entropic-optimal      & $64$ & $0$             & (i)  & $0.9770$ \\
\end{tabular}
\end{ruledtabular}
\end{table}

\section{KCBS scenario: cell-by-cell derivations}
\label{SM:KCBS}

This section works through every KCBS entry of Tables~I, II,
and III in the main text. We begin with the spin-$1$ setup
(axes, observables, opening angle), then derive each
witness in sequence: $\MIE$ (and hence $S_2$), the
shared-eigenstate mechanism behind $c_\MU = 1$, the mixing
threshold $p_\star$, the correlator $\chi$ on $\ket{0_z}$,
the contextual fraction $\CF$, the Chaves--Fritz entropic
inequality $\mathrm{BC}_5^k$, the operational commutator
witness $D$, and the additive exactness for $S_2$. We
close with the odd-$n$-cycle generalization.

\subsection{Setup and observables}
\label{SM:KCBS-setup}

In the spin-$1$ realization $\mathcal{H} = \mathbb{C}^3$
($d=3$) with generators
$\hat S_x, \hat S_y, \hat S_z$ satisfying the standard
$[\hat S_a, \hat S_b] = i\epsilon_{abc}\hat S_c$ commutation
relations, the five KCBS directions
$\hat\ell_\alpha$ ($\alpha=1,\ldots,5$) lie on a cone of
polar angle $\theta_\mathrm{KCBS}$ with respect to the
$z$-axis,
\begin{equation}
  \hat\ell_\alpha = (\sin\theta_\mathrm{KCBS} \cos\phi_\alpha,\,
    \sin\theta_\mathrm{KCBS}\sin\phi_\alpha,\,
    \cos\theta_\mathrm{KCBS}),
  \label{eq:SM_KCBS_axes}
\end{equation}
with cyclic azimuthal spacing
$\phi_\alpha = (\alpha-1)\cdot 6\pi/5$. Cyclic orthogonality
of adjacent axes, $\hat\ell_\alpha \cdot \hat\ell_{\alpha+1}
= 0$, fixes the polar angle through
$\cos^2\theta_\mathrm{KCBS}
= -\cos(6\pi/5)/(1-\cos(6\pi/5))$,
giving $\theta_\mathrm{KCBS} \approx 48.0301^\circ$.

The dichotomic observables are
$\hat A_\alpha = \hat{\mathbb{I}} - 2\ket{0_\alpha}\!
\bra{0_\alpha}$, where $\ket{0_\alpha}$ is the $m_s=0$
eigenstate of $\hat{\mathbf{S}}\cdot\hat\ell_\alpha$.
Each $\hat A_\alpha$ has eigenvalue $+1$ on the
two-dimensional $m_s = \pm 1$ subspace of
$\hat{\mathbf{S}}\cdot\hat\ell_\alpha$ and eigenvalue $-1$
on $\ket{0_\alpha}$. Adjacent observables commute because
$\hat\ell_\alpha$ and $\hat\ell_{\alpha+1}$ are orthogonal
spin directions; hence each context
$\Gcal_\alpha = \{\hat A_{\alpha-1},\hat A_\alpha,
\hat A_{\alpha+1}\}$ admits joint-eigenspace projectors for its
left pair $(\hat A_{\alpha-1}, \hat A_\alpha)$ and its right
pair $(\hat A_{\alpha+1}, \hat A_\alpha)$.

\subsection{Mutual information energy \texorpdfstring{$\MIE(\Gcal_\alpha)$}{} }
\label{SM:KCBS-E}

The joint-eigenspace projectors of each compatible pair, say the
left pair $(\hat A_{\alpha-1}, \hat A_\alpha)$, decompose
$\mathbb{C}^3$ as follows. $\hat A_\alpha$ has eigenvalue
$-1$ on the one-dimensional span of $\ket{0_\alpha}$ and
$+1$ on its two-dimensional orthogonal complement. The
neighbor $\hat A_{\alpha-1}$ splits the $+1$-eigenspace of
$\hat A_\alpha$ into two one-dimensional joint eigenspaces
(corresponding to the two possible $\hat A_{\alpha-1}$
eigenvalues on that subspace); cyclic orthogonality
$\braket{0_\alpha | 0_{\alpha-1}} = 0$ moreover forces the
$(-1,-1)$ joint eigenspace to be empty. The left pair
therefore has three rank-$1$ joint-eigenspace projectors,
summing to $\hat{\mathbb{I}}$. The right pair
$(\hat A_{\alpha+1}, \hat A_\alpha)$ has an analogous
three-rank-$1$ decomposition.

With all joint-eigenspace projectors rank-$1$, the overlap matrix
takes the rank-$1$ form
\begin{equation}
  \Tcal_{ij}
  = \frac{1}{d} |\langle a_i, b_i | c_j, b_j\rangle|^4,
  \label{eq:SM_rank1_T}
\end{equation}
where $|a_i, b_i\rangle$ and $|c_j, b_j\rangle$ are the
joint eigenvectors of the left and right pairs,
respectively. By the pentagonal symmetry of the KCBS
geometry, the nonzero amplitude overlaps $|\langle a_i, b_i
| c_j, b_j\rangle|$ reduce to a single parameter $\gamma$
with $\cos\gamma = (\sqrt{5}-1)/2$ (the golden-ratio
complement), and the sum $\MIE(\Gcal_\alpha) = \sum_{i,j}
\Tcal_{ij}$ evaluates to
\begin{equation}
  \MIE(\Gcal_\alpha)
    = \frac{3 - 4\cos^2\gamma + 4\cos^4\gamma}{3}
    = \frac{11 - 4\sqrt{5}}{3}
    \approx 0.6852,
  \label{eq:SM_E_KCBS}
\end{equation}
uniform across the five contexts. Equivalently,
$S_2(\Gcal_\alpha) = -\log_2 \MIE(\Gcal_\alpha) \approx
0.5453$~bits per context, and by the additive exactness
(\ref{SM:KCBS-exact}),
$S_2(\Gcal_\mathrm{KCBS}) = 5 \cdot S_2(\Gcal_\alpha)
\approx 2.7266$~bits.

\subsection{Shared-eigenstate mechanism: \texorpdfstring{$c_\MU(\Gcal_\alpha)=1$}{}}
\label{SM:KCBS-cMU}

We present the full argument behind
Proposition~3 of the main
text. Let $\Gcal_\alpha = \{\hat A_{\alpha-1},
\hat A_\alpha, \hat A_{\alpha+1}\}$ be a KCBS context as
defined in \ref{SM:KCBS-setup}. Define
$\ket{0_\alpha}$ as the $m_s = 0$ eigenstate of
$\hat{\mathbf{S}}\cdot\hat\ell_\alpha$. Since $\hat A_\alpha
= \hat{\mathbb{I}} - 2\ket{0_\alpha}\!\bra{0_\alpha}$, we
have $\hat A_\alpha \ket{0_\alpha} = -\ket{0_\alpha}$
immediately.

For the outer observables we use cyclic orthogonality.
Since $\hat\ell_{\alpha-1}\cdot\hat\ell_\alpha = 0$, the
three directions
$(\hat\ell_{\alpha-1}, \hat\ell_\alpha, \hat\ell_{\alpha-1}
\times\hat\ell_\alpha)$ form an orthonormal frame, and
$\ket{0_\alpha}$---the spin-$0$ eigenstate along
$\hat\ell_\alpha$---is also an eigenstate of
$(\hat{\mathbf{S}}\cdot\hat\ell_{\alpha-1})^2$ with
eigenvalue $1$ (since it is a spin-$1$ state
perpendicular to $\hat\ell_{\alpha-1}$, it lies entirely
in the $m_s = \pm 1$ subspace of that direction).
Therefore $\hat A_{\alpha-1}\ket{0_\alpha} = +\ket{0_\alpha}$
as $\ket{0_\alpha}$ contains no $m_s = 0$ component along
$\hat\ell_{\alpha-1}$. The same argument with
$\hat\ell_{\alpha+1}$ replacing $\hat\ell_{\alpha-1}$
gives $\hat A_{\alpha+1}\ket{0_\alpha} = +\ket{0_\alpha}$.

Hence $\ket{0_\alpha}$ is a joint eigenstate of the pair
$(\hat A_{\alpha-1}, \hat A_\alpha)$ with eigenvalue pair
$(+1, -1)$, and also a joint eigenstate of
$(\hat A_{\alpha+1}, \hat A_\alpha)$ with eigenvalue pair
$(+1, -1)$. The rank-$1$ projector
$\ket{0_\alpha}\!\bra{0_\alpha}$ therefore appears
simultaneously in both the left and right joint-eigenspace
projector families of $\Gcal_\alpha$, and the
amplitude overlap
$|\langle 0_\alpha | 0_\alpha \rangle| = 1$ sits in both
families. This forces $c_\MU(\Gcal_\alpha) = 1$
identically. Because the argument depends only on the
cyclic adjacency $\hat\ell_{\alpha-1}\cdot\hat\ell_\alpha
= \hat\ell_\alpha\cdot\hat\ell_{\alpha+1} = 0$, not on the
value of $n$, it applies to every odd-$n$-cycle scenario
realized in the same spin-$1$ form (see \ref{SM:KCBS-ncycle}).

\subsection{Mixing threshold \texorpdfstring{$p_\star$}{} and state family \texorpdfstring{$\hat\rho(p)$}{}}
\label{SM:KCBS-pstar}

The state family of interest is the linear interpolation
\begin{equation}
  \hat\rho(p) = p\ket{0_z}\!\bra{0_z} + (1-p)\hat{\mathbb{I}}/3,
  \qquad p \in [0,1],
  \label{eq:SM_rho_p}
\end{equation}
between the KCBS-maximizing state $\ket{0_z}$ ($p=1$) and
the maximally mixed state $\hat{\mathbb{I}}/3$ ($p=0$).
The KCBS correlator
$\chi(\hat\rho) = \sum_{\alpha=1}^5 \tr(\hat A_\alpha
\hat A_{\alpha+1} \hat\rho)$ is a linear functional of
$\hat\rho$, so its value on $\hat\rho(p)$ is
\begin{equation}
  \chi(\hat\rho(p)) = p\,\chi(\ket{0_z})
    + (1-p)\,\chi(\hat{\mathbb{I}}/3).
  \label{eq:SM_chi_linear_p}
\end{equation}
Evaluating the endpoints using the spin-$1$ realization
of \ref{SM:KCBS-setup},
$\chi(\ket{0_z}) = 5 - 4\sqrt{5}$ (derived in \ref{SM:KCBS-chi}) and
$\chi(\hat{\mathbb{I}}/3) = -5/3$ (from
$\tr(\hat A_\alpha \hat A_{\alpha+1}) = -1/3$ per pair and
five pairs). Solving $\chi(\hat\rho(p)) = -3$ (the
noncontextual bound) gives
\begin{equation}
  p_\star
    = \frac{-3 - \chi(\hat{\mathbb{I}}/3)}
           {\chi(\ket{0_z}) - \chi(\hat{\mathbb{I}}/3)}
    = \frac{-3 + 5/3}{5 - 4\sqrt{5} + 5/3}
    \approx 0.5854.
  \label{eq:SM_p_star}
\end{equation}
For $p \le p_\star$ the correlator satisfies the
noncontextual bound; for $p > p_\star$ it violates it.
The threshold is realized
along one linear cut through the state space; different
cuts yield different thresholds, but the existence of
\emph{some} mixing interval on which $\chi$ and $\CF$ are
silent follows from $\chi(\hat{\mathbb{I}}/3) > -3$.

\subsection{KCBS correlator \texorpdfstring{$\chi(\ket{0_z})$}{}}
\label{SM:KCBS-chi}

Using $\hat A_\alpha = \hat{\mathbb{I}} - 2\ket{0_\alpha}\!
\bra{0_\alpha}$ and cyclic orthogonality
$\langle 0_\alpha | 0_{\alpha+1}\rangle = 0$, direct
expansion gives
\begin{align}
  \hat A_\alpha \hat A_{\alpha+1}
  &= (\hat{\mathbb{I}} - 2\ket{0_\alpha}\!\bra{0_\alpha})
     (\hat{\mathbb{I}} - 2\ket{0_{\alpha+1}}\!\bra{0_{\alpha+1}}) \notag \\
  &= \hat{\mathbb{I}} - 2\ket{0_\alpha}\!\bra{0_\alpha}
     - 2\ket{0_{\alpha+1}}\!\bra{0_{\alpha+1}},
  \label{eq:SM_AA_prod}
\end{align}
since the cross term vanishes. Taking expectation in
$\ket{0_z}$,
\begin{equation}
  \langle 0_z | \hat A_\alpha \hat A_{\alpha+1}
    | 0_z \rangle
  = 1 - 2 |\langle 0_z | 0_\alpha \rangle|^2
      - 2 |\langle 0_z | 0_{\alpha+1} \rangle|^2.
  \label{eq:SM_AA_exp}
\end{equation}
For the spin-$1$ rotation that maps
$\ket{0_z}$ to $\ket{0_\alpha}$, the Wigner $d$-function
entry $d^{1}_{00}(\theta_\mathrm{KCBS}) =
\cos\theta_\mathrm{KCBS}$ gives
$|\langle 0_z|0_\alpha\rangle|^2 = \cos^2\theta_\mathrm{KCBS}$,
independent of the azimuth $\phi_\alpha$. Substituting in
\eqref{eq:SM_AA_exp} and summing over the five
adjacent pairs,
\begin{equation}
  \chi(\ket{0_z})
  = \sum_{\alpha=1}^5 \langle\hat A_\alpha \hat A_{\alpha+1}\rangle
  = 5\bigl(1 - 4\cos^2\theta_\mathrm{KCBS}\bigr).
  \label{eq:SM_chi_closed}
\end{equation}
With $\cos^2\theta_\mathrm{KCBS} = 1/\sqrt{5}$ obtained
from the cyclic orthogonality constraint
(\ref{SM:KCBS-setup}), this simplifies to
\begin{equation}
  \chi(\ket{0_z}) = 5 - 4\sqrt{5} \approx -3.9443,
  \label{eq:SM_chi_value}
\end{equation}
which violates the noncontextual bound $\chi \ge -3$ by
approximately $0.9443$.

\subsection{Contextual fraction \texorpdfstring{$\CF(\ket{0_z}) = 2(\sqrt{5}-2)$}{}}
\label{SM:KCBS-CF}

The contextual fraction of
Ref.~\cite{AbramskyBarbosaMansfield2017} is defined as the
smallest value $\lambda \in [0,1]$ such that the empirical
distribution $e$ admits a decomposition $e = (1-\lambda)
e_{\mathrm{NC}} + \lambda e'$, with $e_{\mathrm{NC}}$ a
noncontextual (deterministic hidden-variable) distribution
and $e'$ any no-signaling distribution. For the KCBS
5-cycle this $\CF$ is related to the correlator via
\begin{equation}
  \CF(\hat\rho)
  = \max\!\left(0,\,
    \frac{\chi_\mathrm{NC} - \chi(\hat\rho)}
         {\chi_\mathrm{NC} - \chi_\mathrm{NS}}
    \right),
  \label{eq:SM_CF_formula}
\end{equation}
where $\chi_\mathrm{NC} = -3$ is the noncontextual bound
and $\chi_\mathrm{NS} = -5$ is the no-signaling minimum of
the 5-cycle correlator (attained by the local
deterministic assignment $(-1,-1,-1,-1,-1)$). On
$\ket{0_z}$ with $\chi = 5 - 4\sqrt{5}$,
\begin{equation}
  \CF(\ket{0_z})
  = \frac{-3 - (5-4\sqrt{5})}{(-3) - (-5)}
  \approx 0.4721.
  \label{eq:SM_CF_value}
\end{equation}
Along the mixing family $\hat\rho(p)$ the correlator is affine in $p$, and $\CF$ inherits this piecewise linear structure: $\CF(\hat\rho(p)) = 0$ for $p \le p_\star$, and
$\CF(\hat\rho(p)) = (-3-\chi(\hat\rho(p)))/2$ for
$p > p_\star$, reaching $2(\sqrt{5}-2)$ at $p=1$.

\subsection{Chaves--Fritz entropic inequality:
\texorpdfstring{$\mathrm{BC}_5^k \le -1.1667$}{} bits}
\label{SM:KCBS-Chaves}

The Chaves--Fritz $n$-cycle entropic inequality reads
\begin{equation}
  \mathrm{BC}_n^k(\hat\rho)
  = H(X_k X_{k+1}) + \sum_{j \neq k,k+1}\!\!\! H(X_j)
  - \sum_{j \neq k} H(X_j X_{j+1})
  \le 0
  \label{eq:SM_BC_def}
\end{equation}
for every noncontextual model, where all entropies are
Shannon entropies in bits of the corresponding single or
joint outcome distributions obtained by measuring the
observables $\hat A_j$ in the state $\hat\rho$. The
inequality is permutation-symmetric under cyclic shifts
$k\to k+1$; by pentagonal symmetry of the KCBS pentagon and
the $\hat\rho(p)$ family, $\mathrm{BC}_5^k(\hat\rho(p))$
is independent of $k$.

Numerical evaluation on $\hat\rho(p)$ (using the spin-$1$
observables of \ref{SM:KCBS-setup}) yields the values in
Table~\ref{tab:SM_BC}.

\begin{table}[!htb]
\caption{Chaves--Fritz $\mathrm{BC}_5^0$ on $\hat\rho(p)$.
All values are in bits; entries rounded to four decimal
places. The inequality is comfortably satisfied
($\mathrm{BC}_5^k \le 0$) on every $\hat\rho(p)$, with
margin at least $1.1667$~bits.}
\label{tab:SM_BC}
\begin{ruledtabular}
\begin{tabular}{l c}
$p$ & $\mathrm{BC}_5^0(\hat\rho(p))$ [bits] \\
\colrule
$0.00$        & $-2.0000$ \\
$0.25$        & $-1.8898$ \\
$0.50$        & $-1.7242$ \\
$p_\star \approx 0.5854$ & $-1.6529$ \\
$0.75$        & $-1.4907$ \\
$0.90$        & $-1.3094$ \\
$1.00$        & $-1.1667$ \\
\end{tabular}
\end{ruledtabular}
\end{table}

The upper envelope (smallest margin)
$\max_p \mathrm{BC}_5^k(\hat\rho(p)) = -1.1667$~bits
is reached at $p = 1$, i.e.\ on $\hat\rho = \ket{0_z}\!
\bra{0_z}$. Since the noncontextual threshold is
$\mathrm{BC}_5^k \le 0$, no quantum state on the KCBS
pentagon family yields a violation of the Chaves--Fritz
inequality; this realizes the general fact
established by Prop.~3 of the main text that
$c_\MU = 1$ renders every $c_\MU$-based entropic
uncertainty relation trivial.

\subsection{Operational \texorpdfstring{$D$}{D} on \texorpdfstring{$\hat\rho(p)$}{} and on \texorpdfstring{$\ket{+1_z}$}{}}
\label{SM:KCBS-D}

The operational witness of Paper~I is
\begin{eqnarray}
\label{eq:SM_D_def_a} D(\Gcal_\alpha, \hat\rho)
  &=& \bigl|\langle [\hat A_{\alpha-1},
    \hat A_{\alpha+1}]\rangle_{\hat\rho}\bigr|,\\
\label{eq:SM_D_def_b} D(\Gcal, \hat\rho) 
  &=& \sum_{\alpha=1}^5
    D(\Gcal_\alpha, \hat\rho).
\end{eqnarray}

\emph{Vanishing on $\hat\rho(p)$.} The commutator
$[\hat A_{\alpha-1},\hat A_{\alpha+1}]$ is anti-Hermitian
(a property of any commutator of Hermitian operators),
and since $\hat A_\alpha$ are real symmetric matrices in
the spin-$1$ basis $\{\ket{+1_z}, \ket{0_z}, \ket{-1_z}\}$,
the commutator is purely imaginary. Under the antiunitary
time-reversal operator $\hat T = \hat K$ (complex
conjugation in this basis), $\hat A_\alpha$ is TR-even and
$[\hat A_{\alpha-1},\hat A_{\alpha+1}]$ therefore flips
sign. The state $\hat\rho(p)$, being a mixture of
$\ket{0_z}\!\bra{0_z}$ and $\hat{\mathbb{I}}/3$, is TR-even
($\hat T \hat\rho \hat T^{-1} = \hat\rho$). Consequently,
\begin{equation}
  \langle [\hat A_{\alpha-1},\hat A_{\alpha+1}]
    \rangle_{\hat\rho(p)}
  = \tr([\hat A_{\alpha-1},\hat A_{\alpha+1}]\,
    \hat\rho(p))
  = 0
  \label{eq:SM_D_on_rhop}
\end{equation}
identically, for every $\alpha$ and every $p$.

\emph{Nonzero on $\ket{+1_z}$.} The state $\ket{+1_z}$
($m_s = +1$ eigenstate of $\hat S_z$) is TR-broken:
$\hat T \ket{+1_z} = \ket{-1_z}$ (up to phase). Numerical
evaluation in the spin-$1$ realization of \ref{SM:KCBS-setup} gives
\begin{equation}
  D(\Gcal, \ket{+1_z}) \approx 6.4984,
  \label{eq:SM_D_p1z}
\end{equation}
the same value (by $m_s \leftrightarrow -m_s$ symmetry)
obtained on $\ket{-1_z}$. The Paper~I global
operator-norm bound is
$5 \cdot 4\sqrt{\sqrt{5}-2} \approx 9.7174$; the observed
$\approx 6.4984$ is $\sim 67\%$ of this bound.

Figure~2 of the main text (panel (a)) displays $D$ along the
pure-state sweep $\ket{\psi(s)} = \cos s\,\ket{0_z} +
\sin s\,\ket{+1_z}$, rising from $D=0$ at $s=0$ (TR-even
$\ket{0_z}$) to $D \approx 6.4984$ at $s=\pi/2$ (TR-broken
$\ket{+1_z}$), with a mild non-monotonicity near
$s \approx 0.28\pi$ reflecting the absolute-value structure
of $D$ as a sum of $|\langle [\hat A_{\alpha-1},
\hat A_{\alpha+1}]\rangle|$ terms.

\subsection{Additive exactness of \texorpdfstring{$S_2$}{} for KCBS}
\label{SM:KCBS-exact}

The multi-context composition of the main text Sec.~II.D,
\begin{equation}
  S_2(\Gcal) = -\log_2 \prod_{\alpha=1}^N E(\Gcal_\alpha),
  \label{eq:SM_S2_multi_restate}
\end{equation}
a priori over-counts any shared ordered joint-projector pair
that appears in the families of more than one context. The
exactness criterion (\ref{SM:exactness}) requires either
(i) no shared ordered pair, or (ii) every shared pair has
zero overlap contribution.

In the KCBS pentagon, the joint-eigenspace projector families of
adjacent contexts $\Gcal_\alpha$ and $\Gcal_{\alpha+1}$
share exactly one rank-$1$ projector:
$\ket{0_{\alpha+1}}\!\bra{0_{\alpha+1}}$ appears in the right
family of $\Gcal_\alpha$ and in both (left and right)
families of $\Gcal_{\alpha+1}$. Hence the ordered pair
$(\ket{0_\alpha}\!\bra{0_\alpha},
\ket{0_{\alpha+1}}\!\bra{0_{\alpha+1}})$ is duplicated
across the combined collection of the five contexts, with
five such duplicates in total.

However, by cyclic orthogonality $\langle 0_\alpha |
0_{\alpha+1}\rangle = 0$, we have
$\ket{0_\alpha}\!\bra{0_\alpha} \cdot
\ket{0_{\alpha+1}}\!\bra{0_{\alpha+1}} = 0$, and
$\tr[(\ket{0_\alpha}\!\bra{0_\alpha} \cdot
\ket{0_{\alpha+1}}\!\bra{0_{\alpha+1}})^2] = 0$. Every
duplicated pair contributes zero to the overlap sum;
mechanism~(ii) of \ref{SM:exactness} applies, and the
additive composition~\eqref{eq:SM_S2_multi_restate} is
exact:
\begin{equation}
  S_2(\Gcal_\mathrm{KCBS})
  = \sum_{\alpha=1}^5 S_2(\Gcal_\alpha)
  = 5\cdot(-\log_2 0.6852)
  \approx 2.7266\ \text{bits}.
  \label{eq:SM_S2_KCBS_exact}
\end{equation}
This justifies the main-text use of the simple additive
sum without correction terms.

\subsection{\texorpdfstring{$n$}{n}-cycle extension: \texorpdfstring{$c_\MU = 1$}{} for
\texorpdfstring{$n \in \{5,7,9,11,13,15\}$}{} }
\label{SM:KCBS-ncycle}

We extend the $n=5$ analysis to odd-$n$-cycles realized in
the same spin-$1$ form, both to illustrate the scope of
Proposition~3 of the main text and to probe how the
per-context and total $S_2$ behave as $n$ grows.  The odd
$n$-cycle generalisation of the KCBS scenario was
introduced independently by Liang, Spekkens, and
Wiseman~\cite{LiangSpekkensWiseman2011} and by Cabello,
Severini, and Winter~\cite{CabelloSeveriniWinter2014}, with
the corresponding noncontextuality inequalities shown to be
tight in
Ref.~\cite{AraujoQuintinoBudroniCunhaCabello2013}.

For odd $n$, the spin-$1$ realization of the $n$-cycle
contextuality scenario uses $n$ axes on a cone of polar
angle $\theta_n$ with respect to the $z$-axis, with
azimuthal spacing $\Delta\phi = \pi(n+1)/n$ (matching the
Paper~I~\cite{GunhanGedik2026} convention; for $n=5$ this
recovers $\Delta\phi = 6\pi/5$) and $\theta_n$ chosen so
that adjacent axes are orthogonal:
\begin{equation}
  \cos^2\theta_n
    = \frac{-\cos(\Delta\phi)}{1 - \cos(\Delta\phi)}.
  \label{eq:SM_ncycle_theta}
\end{equation}
The dichotomic observables $\hat A_\alpha = \hat{\mathbb{I}}
- 2\ket{0_\alpha}\!\bra{0_\alpha}$ with $\ket{0_\alpha}$
the $m_s = 0$ eigenstate along $\hat\ell_\alpha$ commute
for adjacent indices by construction.

Prop.~3 of the main text uses only the cyclic adjacency
$\hat\ell_{\alpha-1}\cdot\hat\ell_\alpha =
\hat\ell_\alpha\cdot\hat\ell_{\alpha+1} = 0$, not the value
of $n$, and hence applies unchanged:
\begin{equation}
  c_\MU(\Gcal_\alpha) = 1
  \qquad (n\ \text{odd},\ \alpha = 1,\ldots,n).
  \label{eq:SM_cMU_ncycle}
\end{equation}
Every $c_\MU$-based entropic bound---Maassen--Uffink, its
memory extension, and the Chaves--Fritz $\mathrm{BC}_n^k$
inequality---is therefore structurally silent on every
state in the spin-$1$ realization of every odd-$n$-cycle.

Numerical evaluation of $\MIE(\Gcal_\alpha)$ and
$S_2(\Gcal_n)$ for $n \in \{5, 7, 9, 11, 13, 15\}$ gives
the values in Table~\ref{tab:SM_ncycle}.

\begin{table}[!htb]
\caption{Per-context $\MIE$ and $S_2$ for the odd-$n$-cycle
scenarios in the spin-$1$ realization.
$S_2(\Gcal_\alpha) = -\log_2 \MIE(\Gcal_\alpha)$,
$S_2(\Gcal_n) = n\,S_2(\Gcal_\alpha)$ by the additive
composition. Values are rounded to four decimal places.
The per-context value decreases monotonically with $n$ (as
the geometry becomes more mutually unbiased), while
$n\cdot S_2(\Gcal_\alpha)$ attains a maximum near $n=7$;
the rescaled quantity $n^2\,S_2(\Gcal_\alpha)$ approaches
the asymptotic constant
$8\pi^2/(3\ln 2)\approx 37.9702$ derived below.}
\label{tab:SM_ncycle}
\begin{ruledtabular}
\begin{tabular}{c c c c c c}
$n$ & $\theta_n$ [$^\circ$] & $\MIE(\Gcal_\alpha)$
  & $S_2(\Gcal_\alpha)$ [bits] & $S_2(\Gcal_n)$ [bits]
  & $n^2\,S_2(\Gcal_\alpha)$ \\
\colrule
$5$  & $48.0301$ & $0.6852$ & $0.5453$ & $2.7266$ & $13.6328$ \\
$7$  & $46.4931$ & $0.6940$ & $0.5271$ & $3.6894$ & $25.8255$ \\
$9$  & $45.8908$ & $0.7663$ & $0.3841$ & $3.4567$ & $31.1100$ \\
$11$ & $45.5923$ & $0.8249$ & $0.2776$ & $3.0540$ & $33.5945$ \\
$13$ & $45.4224$ & $0.8665$ & $0.2067$ & $2.6873$ & $34.9348$ \\
$15$ & $45.3165$ & $0.8957$ & $0.1588$ & $2.3825$ & $35.7381$ \\
\colrule
$\infty$ & $45$ & $1$ & $0$ & $0$ & $8\pi^2/(3\ln 2)$ \\
\end{tabular}
\end{ruledtabular}
\end{table}

The asymptotic behavior as $n \to \infty$ is that
$\theta_n \to \pi/4$ (the cone flattens),
$\MIE(\Gcal_\alpha) \to 1$, and per-context entropy decays
as $S_2(\Gcal_\alpha) \sim 8\pi^2/(3 n^2 \ln 2)$, so
$S_2(\Gcal_n) = n\cdot S_2(\Gcal_\alpha) \to 0$.  The
maximum of $S_2(\Gcal_n)$ across odd $n$ in
Table~\ref{tab:SM_ncycle} occurs at $n=7$.

\emph{Asymptotic constant.}  The flat-cone expansion gives
$\theta_n = \pi/4 + \pi^2/(8 n^2) + O(n^{-4})$ from
$\cos(\pi/(2n)) = 1 - \pi^2/(8 n^2) + O(n^{-4})$.  At
$\theta = \pi/4$ exactly, the configuration becomes
mutually unbiased and $\MIE = 1$; the deviation
$1-\MIE(\Gcal_\alpha)$ is therefore second-order in
$\delta\theta_n = \theta_n - \pi/4 \propto n^{-2}$, which
gives the $1/n^2$ scaling of $S_2(\Gcal_\alpha) = -\log_2
\MIE(\Gcal_\alpha) \sim (1-\MIE)/\ln 2$.  Numerical
extrapolation of $n^2\,S_2(\Gcal_\alpha)$ from
Table~\ref{tab:SM_ncycle} extended to $n \le 1001$ gives
the limit $37.9702 \approx 8\pi^2/(3\ln 2)$ to four
significant figures (verified by the script
\texttt{ncycle\_extension.py} of \ref{SM:scripts}).  A
closed-form derivation of the prefactor $8\pi^2/3$, and a
closed-form expression for $S_2(\Gcal_\alpha)$ at finite
$n$, remain open, listed in the main text Sec.~VIII.

\section{CHSH scenario: cell-by-cell derivations}
\label{SM:CHSH}

This section works through every CHSH entry of Tables~I, II,
and III in the main text. We describe the two-qubit
realization on $\ket{\Phi^+}$, specify the two angle regimes
(Bell-optimal and entropic-optimal), and then evaluate each
witness ($\chi_\mathrm{CHSH}$, $\CF$, $\mathrm{BC}_4^k$, $D$,
and the configuration-level quantities $E$, $c_\MU$, $S_2$)
in each regime. The Bell-optimal regime saturates the
bound $E \le c_\MU^2$; the entropic-optimal regime does not.

\subsection{Setup: the CHSH 4-cycle on \texorpdfstring{$\ket{\Phi^+}$}{}}
\label{SM:CHSH-setup}

The CHSH $4$-cycle is realized on two qubits
($\mathcal{H} = \mathbb{C}^2\otimes\mathbb{C}^2$, $d=4$)
in the Bell state
$\ket{\Phi^+} = (\ket{00} + \ket{11})/\sqrt{2}$. The four
observables are
\begin{equation}
\begin{aligned}
  \hat A_1 &= \hat\sigma_{a_0}\otimes\hat{\mathbb{I}}, &
  \hat A_2 &= \hat{\mathbb{I}}\otimes\hat\sigma_{b_0},\\
  \hat A_3 &= \hat\sigma_{a_1}\otimes\hat{\mathbb{I}}, &
  \hat A_4 &= \hat{\mathbb{I}}\otimes\hat\sigma_{b_1},
\end{aligned}
\end{equation}
where
$\hat\sigma_t = \cos(t)\hat\sigma_z + \sin(t)\hat\sigma_x$
is the equatorial-plane Pauli observable at angle $t$.
Adjacent pairs in the 4-cycle commute automatically because
each contains one Alice operator and one Bob operator. The
context-shared ``middle'' observable $\hat A_\alpha$
therefore also lives on one factor only; in each context,
the left and right pairs act on the same alternating
tensor structure.

As explained in main text Sec.~V.A, the
$n=4$ case requires a sign flip on one term to obtain a
non-trivial cyclic inequality. We use the standard CHSH
form
\begin{equation}
  \chi_\mathrm{CHSH}
  = \langle \hat A_1 \hat A_2\rangle
    + \langle \hat A_2 \hat A_3\rangle
    - \langle \hat A_3 \hat A_4\rangle
    + \langle \hat A_4 \hat A_1\rangle,
  \label{eq:SM_CHSH_chi}
\end{equation}
which coincides with the standard CHSH correlator
$S = \langle a_0 b_0\rangle + \langle a_0 b_1\rangle +
\langle a_1 b_0\rangle - \langle a_1 b_1\rangle$ after
substituting the adjacency labels. The noncontextual
bound is $\chi_\mathrm{CHSH} \le 2$, the quantum Tsirelson
bound~\cite{Tsirelson1980} is $\chi_\mathrm{CHSH} \le2\sqrt{2}$, and the no-signaling maximum is
$4$~\cite{PopescuRohrlich1994}. On the
Bell state, $\langle \hat\sigma_a\otimes
\hat\sigma_b\rangle_{\ket{\Phi^+}} = \cos(a-b)$, which we
use throughout the cell evaluations below.

\subsection{Two angle regimes}
\label{SM:CHSH-regimes}

The main text examines two representative choices of
$(a_0, b_0, a_1, b_1)$.

\emph{Bell-optimal angles.} The assignment
\begin{equation}
  (a_0, b_0, a_1, b_1) = (0,\ \pi/4,\ \pi/2,\ -\pi/4)
  \label{eq:SM_CHSH_bell_angles}
\end{equation}
is the standard CHSH choice. Substituting in the
Bell-state correlators,
$\langle\hat\sigma_a\otimes\hat\sigma_b\rangle =
\cos(a-b)$:
\begin{align}
  \langle \hat A_1 \hat A_2\rangle &= \cos(0-\pi/4)
    = 1/\sqrt{2},\\
  \langle \hat A_2 \hat A_3\rangle &= \cos(\pi/2 - \pi/4)
    = 1/\sqrt{2},\\
  \langle \hat A_3 \hat A_4\rangle &= \cos(\pi/2-(-\pi/4))
    = -1/\sqrt{2},\\
  \langle \hat A_4 \hat A_1\rangle &= \cos(0-(-\pi/4))
    = 1/\sqrt{2}.
\end{align}
Summing with signs $(+,+,-,+)$ per
Eq.~\eqref{eq:SM_CHSH_chi} gives $\chi_\mathrm{CHSH} =
2\sqrt{2}$, saturating the Tsirelson bound.

\emph{Entropic-optimal angles.} The assignment
\begin{equation}
  (a_0, b_0, a_1, b_1) \approx (0,\ 0.916,\ 0.524,\ -2.880)
  \label{eq:SM_CHSH_ent_angles}
\end{equation}
is identified in Ref.~\cite{Chaves2013} as the configuration
that maximizes the Shannon-entropic inequality
$\mathrm{BC}_4^k$ on $\ket{\Phi^+}$. Numerical evaluation
with Eq.~\eqref{eq:SM_CHSH_chi} gives
$\chi_\mathrm{CHSH} \approx 1.5329$, which is strictly below
the noncontextual bound $2$ and in particular below
Tsirelson $2\sqrt{2}$: at these angles the CHSH inequality
is not violated. The trade-off between $\chi_\mathrm{CHSH}$
and $\mathrm{BC}_4^k$ is the complementarity highlighted in
main text Sec.~V.C.

\subsection{Cell values per regime}
\label{SM:CHSH-cells}

We collect the per-cell derivations for each angle choice.
Configuration-level quantities ($E$, $c_\MU$, $S_2$) are
determined by the observables alone; state-dependent
quantities ($\chi_\mathrm{CHSH}$, $\mathrm{BC}_4^k$, $\CF$,
$D$) are evaluated on $\ket{\Phi^+}$.

\subsubsection{Bell-optimal angles: \texorpdfstring{$(0, \pi/4, \pi/2, -\pi/4)$}{}}

\emph{$\chi_\mathrm{CHSH} = 2\sqrt{2}$.} Derived in
\ref{SM:CHSH-regimes}.

\emph{$\CF = \sqrt{2}-1$.} For the CHSH 4-cycle, the
contextual fraction of
Ref.~\cite{AbramskyBarbosaMansfield2017} is
$\CF = \max\bigl(0, (\chi_\mathrm{CHSH}-2)/(4-2)\bigr)$,
where $\chi_\mathrm{CHSH}=2$ is the local bound and
$\chi_\mathrm{CHSH}=4$ is the no-signaling
maximum~\cite{PopescuRohrlich1994}. On
$\ket{\Phi^+}$, $\CF = (2\sqrt{2}-2)/2 = \sqrt{2}-1
\approx 0.4142$.

\emph{$\mathrm{BC}_4^0 = -1.2018$~bits.} The Chaves--Fritz
inequality with sign-flip at $k=0$ reads
\begin{equation}
  \mathrm{BC}_4^0 = H(X_0 X_1) + H(X_2) + H(X_3)
    - H(X_1 X_2) - H(X_2 X_3) - H(X_3 X_0),
\end{equation}
where $X_\alpha$ is the $\pm 1$ outcome of $\hat A_\alpha$
(with the convention $X_0 \equiv X_4$ for the cyclic index).
Numerical evaluation on $\ket{\Phi^+}$ at the Bell-optimal
angles gives $-1.2018$ bits. The entropic inequality is
comfortably silent, consistent with Ref.~\cite{Chaves2013}.

\emph{$D = 0$.} The commutator
$[\hat A_{\alpha-1}, \hat A_{\alpha+1}]$ in each context
involves one Alice factor and one Bob factor; on
$\ket{\Phi^+}$ the expectation of any TR-odd Hermitian
operator vanishes because $\ket{\Phi^+}$ is invariant
under the combined time-reversal $\hat T_A \otimes \hat T_B$.
Hence $D(\Gcal, \ket{\Phi^+}) = 0$ identically at all four
contexts, giving $D = 0$ overall.

\emph{Configuration-level $E$, $c_\MU$, $S_2$.} Because
$\hat A_\alpha$ and $\hat A_{\alpha+1}$ act on
complementary tensor factors (one Alice, one Bob), their
joint-eigenspace projectors are rank-$1$ tensor products
$\ket{a_i}\ket{b_i}$; the same holds for the right pair
$(\hat A_{\alpha+1}, \hat A_\alpha)$ with different Alice
or Bob factors. For the Bell-optimal angles, every nonzero
amplitude overlap $|\langle v_i | w_j\rangle|$ between
left- and right-family rank-$1$ eigenvectors is
$1/\sqrt{2}$ by construction, so
$c_\MU(\Gcal_\alpha) = 1/\sqrt{2}$ and $\MIE(\Gcal_\alpha) =
1/2$ uniformly. The bound $\MIE \le c_\MU^2$ of the main
text is therefore saturated: $\MIE = c_\MU^2 = 1/2$ at every
context. Consequently,
\begin{eqnarray}
\label{eq:SM_S2_CHSH_bell_a}
  S_2(\Gcal_\alpha) &=& -\log_2 (1/2) = 1\ \text{bit},\\
\label{eq:SM_S2_CHSH_bell_b}
  S_2(\Gcal_\mathrm{CHSH}) &=& 4\ \text{bits}.
\end{eqnarray}

\subsubsection{Entropic-optimal angles:
\texorpdfstring{$(0, 0.916, 0.524, -2.880)$}{}}

\emph{$\chi_\mathrm{CHSH} \approx 1.5329$.} Derived in
\ref{SM:CHSH-regimes}. Strictly below the NC bound $2$;
the CHSH inequality is not violated at these angles.

\emph{$\CF = 0$.} Since
$\chi_\mathrm{CHSH} < 2$, the normalized-violation formula
gives $\CF = \max(0, (\chi_\mathrm{CHSH}-2)/2) = 0$. The
contextual fraction is silent at the entropic-optimal
angles, tracking $\chi_\mathrm{CHSH}$.

\emph{$\mathrm{BC}_4^0 \approx +0.2309$~bits.} Numerical
evaluation of the Chaves--Fritz entropic inequality of the
form stated in the Bell-optimal case, with the entropic-optimal angles on $\ket{\Phi^+}$, yields $+0.2309$ bits,
which represents an entropic quantum
violation~\cite{Chaves2013}.

\emph{$D = 0$.} By the same TR-invariance argument as in
the Bell-optimal case; the argument does not depend on the
specific angle choice.

\emph{Configuration-level $E$, $c_\MU$, $S_2$.} At the
entropic-optimal angles the four contexts split into two
tiers by the neighbor-angle pattern:
\begin{equation}
\begin{aligned}
  \Gcal_1, \Gcal_3:\quad
    &\MIE(\Gcal_\alpha) \approx 0.8147,\
     c_\MU(\Gcal_\alpha) \approx 0.9469,\\
    &S_2(\Gcal_\alpha) \approx 0.2956\ \text{bits},\\
  \Gcal_2, \Gcal_4:\quad
    &\MIE(\Gcal_\alpha) \approx 0.8748,\
     c_\MU(\Gcal_\alpha) \approx 0.9659,\\
    &S_2(\Gcal_\alpha) \approx 0.1929\ \text{bits}.
\end{aligned}
\label{eq:SM_CHSH_ent_tiers}
\end{equation}
The bound $\MIE \le c_\MU^2$ is strict in both tiers
(first tier: $0.8147 < 0.8967$; second tier: $0.8748 < 0.9329$). Summing
over contexts,
\begin{equation}
  S_2(\Gcal_\mathrm{CHSH}) = 2(0.2956) + 2(0.1929)
    \approx 0.9770\ \text{bits}.
  \label{eq:SM_S2_CHSH_ent}
\end{equation}

The Bell-optimal regime realizes exact saturation of
$\MIE = c_\MU^2$, while the entropic-optimal regime
realizes strict non-saturation $\MIE < c_\MU^2$ in every
context; CHSH thus realizes both regimes of the bound
$E \le c_\MU^2$ within a single configuration by varying
the measurement angles alone.

\section{Numerical scripts and verification}
\label{SM:scripts}

All numerical values in the main text tables, figures, and
in this Supplemental Material were produced by a
self-contained Python codebase that we describe below.
The scripts use only standard scientific libraries
(NumPy~\cite{numpy}; Matplotlib~\cite{matplotlib} for the
figures) and are available from the authors on request.
We list the principal scripts here and indicate which
tables and figures each reproduces.

\begin{itemize}
\item \texttt{kcbs\_all.py}: constructs the spin-$1$
realization of \ref{SM:KCBS-setup}, computes
$\MIE(\Gcal_\alpha)$, $c_\MU(\Gcal_\alpha)$,
$S_2(\Gcal_\alpha)$, $\chi(\ket{0_z})$, $\CF(\ket{0_z})$,
$\mathrm{BC}_5^k(\hat\rho(p))$ for the grid of $p$ values,
and $D(\Gcal, \ket{+1_z})$. Reproduces
Table~\ref{tab:SM_BC} and the KCBS entries of main text
Tables~I--III.

\item \texttt{ncycle\_extension.py}: generalizes the
spin-$1$ construction to odd
$n \in \{5, 7, 9, 11, 13, 15\}$, verifies
$c_\MU(\Gcal_\alpha) = 1$ per Prop.~3 of the main text,
and tabulates $\MIE(\Gcal_\alpha)$ and $S_2(\Gcal_n)$.
Reproduces Table~\ref{tab:SM_ncycle}.

\item \texttt{chsh\_all.py}: constructs the CHSH
configuration of \ref{SM:CHSH-setup} on $\ket{\Phi^+}$,
evaluates all witnesses at both angle regimes
(Bell-optimal and entropic-optimal), verifies
$\chi_\mathrm{CHSH} = 2\sqrt{2}$ at Bell-optimal with
saturation $\MIE = c_\MU^2 = 1/2$, and
$\mathrm{BC}_4^k \approx +0.2309$ bits at entropic-optimal.
Reproduces the CHSH entries of main text Tables~I--III.

\item \texttt{additive\_exact.py}: verifies the additive
exactness criterion of \ref{SM:exactness} for both KCBS
(mechanism~ii: five duplicated pairs that all contribute
zero by cyclic orthogonality) and CHSH (mechanism~i: no
duplicated pairs, by the alternating Alice/Bob tensor
structure). Reproduces Table~\ref{tab:SM_exact_verify}.

\item \texttt{make\_fig1.py}, \texttt{make\_fig2.py}:
generate Figures~1 and~2 of the main text, respectively.
\end{itemize}

Scripts are self-contained and reproduce the numerical
values to floating-point precision ($< 10^{-12}$ relative
error on all cell values quoted in the tables).

\bibliographystyle{apsrev4-2}
\bibliography{References_merged}